\documentclass[a4paper,UKenglish,cleveref, autoref, thm-restate,authorcolumns]{lipics-v2019}
\usepackage[parfill]{parskip}    
\usepackage{verbatim}
\usepackage{graphicx}
\usepackage{epstopdf}
\usepackage{bussproofs}
\usepackage{color} 
\usepackage{dsfont}
\usepackage{stmaryrd}
\usepackage{amsthm}
\usepackage{mathtools}
\usepackage{amssymb}
\usepackage{hyperref}
\usepackage{enumerate}
\usepackage{xcolor}
\usepackage{todonotes}
\usepackage{framed}
\usepackage{comment}
\definecolor{shadecolor}{gray}{0.95}

\usepackage{faktor} % to display quotient properly

\specialcomment{ndr}{\begin{shaded*}}{\end{shaded*}}

\usepackage{scalerel}

\newcommand*{\ntodo}[1]{\hspace{1pt}\newline\todo[color=green,inline]{#1}}

\newcommand{\distr}{\Delta}
\newcommand{\ddistr}{\Phi}
\newcommand{\distrb}{\Theta}

\newcommand{\terms}[2]{\mathcal T(#1,#2)}

\newcommand{\ics}{\kappa}

\newcommand{\alga}{\mathbb A}
\newcommand{\algb}{\mathbb B}

\newcommand{\qt}{\mathcal U}
\newcommand{\qet}{\mathtt{QTh}}
\newcommand{\et}{\mathtt{Th}}

\newcommand{\sigsl}{\Sigma_{SL}}
\newcommand{\sigca}{\Sigma_{CA}}
\newcommand{\sigcs}{\Sigma_{CS}}

\newcommand{\pplus}[1]{+_{#1}}
\DeclareMathOperator*{\bigplus}{\scalerel*{+}{\textstyle{\sum}}}

\newcommand{\dirac}[1]{\delta(#1)}
\newcommand{\bindistr}[3]{#1 \cdot  #2 + (1-#1)  \cdot #3}
\newcommand{\bindistrwms}[3]{\wms( #1 \, #2 + (1-#1) \, #3)}

\newcommand{\haus}{H}
\newcommand{\kant}{K}
\newcommand{\hk}{\haus\kant}

\newcommand{\fpset}{\mathcal P}
\newcommand{\upset}{{\mathcal{P}_u}}
\newcommand{\pset}{\mathcal P}
\newcommand{\dset}{\mathcal D}
\newcommand{\cset}{\mathcal C}

\newcommand{\lcset}{\hat{\cset}}
\newcommand{\lfpset}{\hat{\fpset}}

\newcommand{\lpset}{\hat{\pset}}
\newcommand{\ldset}{\hat{\dset}}

\newcommand{\EM}{\mathbf E\mathbf M}

\newcommand{\Cat}{\mathbf C}

\newcommand{\Sets}{\mathbf{Set}}
\newcommand{\Met}{\mathbf{Met}}

\newcommand{\wms}{\textnormal{\texttt{WMS}}}

\newcommand{\kompactd}[2]{\mathtt{Comp}(#1,#2)}
\newcommand{\hs}{\mathcal V}

\newcommand{\conv}{cc}

\newcommand{\support}{supp}

\newcommand{\mon}{\mathcal M}

\newcommand{\set}[2]{\{#1|\, #2\}}
\newcommand{\setin}[2]{\{#1|\, #1 \in #2\}}
\newcommand{\pushf}[1]{\overline{#1}}

\newcommand{\ema}{Eilenberg-Moore\ }

\newcommand{\funpair}[2]{\langle #1, #2\rangle}
\newcommand{\funleq}{\sqsubseteq}

\newcommand{\calf}{\mathcal F}
\newcommand{\calg}{\mathcal G}

\newcommand{\ub}{\mathtt{UB}} % unique base

\newcommand{\ubeq}{\stackrel{\ub}{=}}

\newcommand{\acat}{\mathbf A}
\newcommand{\qacat}{\mathbf {QA}}

\newcommand{\cplus}{\oplus} %binary choice operator for convex semilattices
\newcommand{\bigcplus}{\bigoplus} %multiple choice for convex semilattices
\newcommand{\bigpplus}{\bigplus} %multiple probabilistic choice for convex semilattices

\newcommand{\etsl}{\et_{SL}}
\newcommand{\etca}{\et_{CA}}
\newcommand{\etcs}{\et_{CS}}

\newcommand{\qetsl}{\qet_{SL}}
\newcommand{\qetca}{\qet_{CA}}
\newcommand{\qetcs}{\qet_{CS}}

\newcommand{\nf}{\nu}

\newcommand{\vvcut}[1]{}

\newcommand{\finaltodo}[1]{}

\title{Monads and Quantitative Equational Theories \\ for  Nondeterminism and Probability}

\author{Matteo Mio}{CNRS \& ENS--Lyon, France}{}{}{}

\author{Valeria Vignudelli}{CNRS \& ENS--Lyon, France}{}{}{}

\authorrunning{M. Mio and V. Vignudelli} %TODO mandatory. First: Use abbreviated first/middle names. Second (only in severe cases): Use first author plus 'et al.'


\nolinenumbers

\begin{document}

\setlength{\abovedisplayskip}{0.1cm}
\setlength{\belowdisplayskip}{0.1cm}

\maketitle

\begin{abstract}
The monad of convex sets of probability distributions is a well--known tool for modelling the combination of nondeterministic and probabilistic computational effects. In this work we lift this monad from the category of sets to the category of metric spaces, by means of the Hausdorff and Kantorovich metric liftings. Our main result is the presentation of this lifted monad in terms of the quantitative equational theory of convex semilattices, using the framework of quantitative algebras  recently introduced by Mardare, Panangaden and Plotkin.
\end{abstract}

\section{Introduction}\label{introduction:section}
%!TEX root = paper.tex

In the theory of programming languages the categorical concept of \emph{monad} is used to handle computational effects \cite{Moggi-89,Moggi-91}. As main examples, the \emph{powerset monad} ($\mathcal{P}$) and the \emph{probability distribution monad} ($\mathcal{D}$) are used to handle nondeterministic and probabilistic behaviours, respectively. It is of course desirable to handle the combination of these two effects to model, for instance, concurrent randomised protocols where nondeterminism arises from the action of an unpredictable scheduler and probability from the use of  randomised procedures such as coin tosses. However, the composite functor $\mathcal{P}\circ\mathcal{D}$ is not a monad (see, e.g., \cite{VW06}).  

A well--known way to handle this technical issue is to use instead the \emph{convex powerset of distributions monad} ($\cset$) which restricts $\mathcal{P} \circ \mathcal{D}$ by only admitting sets of probability distributions that are closed under the formation of \emph{convex combinations} (see  \cite{DBLP:journals/entcs/TixKP09a,DBLP:conf/fossacs/Goubault-Larrecq08a,DBLP:conf/csl/Goubault-Larrecq07,DBLP:journals/entcs/Mislove06,Mislove00,Jacobs08} and Section \ref{section:2}). Restricting $\mathcal{P}\circ\mathcal{D}$ to $\cset$ is not only mathematically convenient, because it leads to a monad, but also natural as convexity captures the possibility of the scheduler to make probabilistic choices, as originally observed by Segala \cite{Seg95:thesis}. Suppose indeed that a scheduler can select between two probabilistic behaviours $\{ d_1, d_2 \}$ for execution. It is reasonable to assume that said scheduler can also, with the aid of a (biased) coin, choose $d_1$ with probability $p$ and $d_2$ with probability $1-p$. Hence, effectively, the scheduler can choose any behaviour in  $\{p\cdot d_1 +  (1-p) \cdot d_2 \mid p\in [0,1] \}$, which is indeed a convex set of distributions. 

In a recent work \cite{BSV19} the authors provide a proof for the following result: the equational theory $\etcs$ of convex semilattices is a \emph{presentation} of the $\Sets$ monad $\cset$. 
%In a recent work \cite{BSV19} the authors provide a proof for the following result: the  $\Sets$  monad $\cset$ is presented by the equational theory $\etcs$ of convex semilattices.
This means (see Section \ref{section:2} for details) that the category $\acat{(\etcs)}$ of convex semilattices and their homomorphisms is isomorphic to the category $\EM(\cset)$ of Eilenberg-Moore algebras for $\cset$. 
%\marginpar{eviterei see section 2 qua, e' gia' sopra. Facciamo una ``structure of the paper'', o mettiamo in vari punti dell'intro ``this is in section blabla''} 
%\marginpar{EM algebras of or for? Monads in Met or monads on Met?}

Presentation results of this kind %provide a bridge between syntax and semantics and 
have a number of applications in computer science due to (quoting Klin \cite[p.1]{DBLP:journals/tcs/Klin11}) the ``interplay between the structure (syntax) and the dynamics (behaviour) of systems.'' 
For example,  it follows from the presentation result of  \cite{BSV19} that the free convex semilattice with set of generators $X$ is isomorphic to $\cset(X)$. This allows us to manipulate elements of $\cset(X)$ as convex semilattice terms modulo the equations of $\etcs$ and, similarly, to perform equational reasoning steps using facts (e.g., from geometry) related to the mathematical structure of $\cset(X)$. 
Applications in the field of program semantics and concurrency theory arise by combining coalgebraic reasoning methods, associated with the use of monads as behaviour functors, and algebraic methods, which are made available by presentation theorems. Well known examples include \emph{bisimulation up--to techniques} (e.g., up--to congruence  \cite{DBLP:journals/acta/BonchiPPR17}) and the categorical approach to structural operational semantics, introduced by Turi and Plotkin in \cite{TP1997} (see also \cite{DBLP:journals/tcs/Klin11}) and based on the notion of \emph{bialgebras}.

%One application is the design of \emph{bisimulation up--to techniques} as, e.g., up--to congruence techniques rely on exploiting algebraic structure of the behaviour functor of choice (see, e.g., \cite{DBLP:journals/acta/BonchiPPR17}). Similarly, the categorical approach to structural operational semantics by Turi and Plotkin \cite{TP1997} (see also \cite{DBLP:journals/tcs/Klin11}) is based on the notion of \emph{bialgebras}. Both concepts are based on the combination of algebraic and coalgebraic methods.
%For example,  it follows from the presentation result of  \cite{BSV19} that the free convex semilattice with set of generators $X$ is isomorphic to $\cset(X)$ and this allows to manipulate elements of $\cset(X)$ as convex semilattice terms modulo the equations of $\etcs$ and, similarly, to perform equational reasoning steps using facts (e.g., from geometry) related to the mathematical structure of $\cset(X)$. 
%This is useful, for example, in the design of \emph{up--to techniques}: e.g., up--to congruence techniques rely on exploiting algebraic structure of the behaviour functor of choice (see, e.g., \cite{DBLP:journals/acta/BonchiPPR17}). Similarly, the categorical approach to structural operational semantics of Turi--Plotkin \cite{TP1997} (see also \cite{DBLP:journals/tcs/Klin11}) is based on the notion of \emph{bialgebras} which are, roughly speaking, coalgebras defined on algebras. 

The category $\Met$, having metric spaces as objects and non--expansive maps as morphisms, is a natural mathematical setting which can replace the category $\Sets$ when it is desirable to switch from the concept of \emph{program equivalence} to that of \emph{program distance}.  This has been a very active topic of research in the last two decades (see, e.g, \cite{prakashbook,GJS90,breugel2005,DJGP02,DBLP:conf/icalp/BreugelW01}). In this context, it is necessary to deal with monads on $\Met$. Variants of the $\Sets$ monads $\pset$ and $\dset$ have been proposed on $\Met$ (see, e.g., \cite{breugel2005,BaldanBKK18} and Section \ref{sec:monadmet}), and are technically based on different types of \emph{metric liftings}, due to Hausdorff and Kantorovich. 
%\marginpar{slightly rephrased. Also, I put powerset instead of finite powerset}
%The category $\Met$, having metric spaces as objects and non--expansive maps as morphisms, is a natural mathematical setting and can be taken as replacement of the category $\Sets$ when it is desirable to switch from the concept of \emph{program equivalence} to that of \emph{program distance}.  This has been a very active topic of research in the last two decades (see, e.g, \cite{prakashbook,GJS90,breugel2005,DJGP02,DBLP:conf/icalp/BreugelW01}). In this context it is necessary to deal with monads on $\Met$. Variants of the $\Sets$ monads $\fpset$ and $\dset$ have been proposed in $\Met$ (see, e.g. \cite{breugel2005,bardan14} and are technically based on different types of \emph{metric liftings}, due to Hausdorff and Kantorovich. 

%These are denoted by $\lfpset$ and $\ldset$ respectively.

\textbf{Contributions of this work.} In this work we investigate  a $\Met$ variant of the $\Sets$ monad $\cset$, which we denote by $\lcset$. As a functor,  $\lcset : \Met \rightarrow \Met$ maps a metric space $(X,d)$ to the metric space $ (\cset(X), \haus\kant(d))$, the collection of non--empty, finitely generated convex sets of finitely supported probability distributions on $X$
 endowed with the metric $\haus(\kant(d))$, the Hausdorff lifting of the Kantorovich lifting of the metric $d$.
$$
\lcset : \Met \rightarrow \Met \ \ \ \ \ \ \ \ \  (X,d) \mapsto  \Big( \cset(X), \haus(\kant(d)) \Big).
$$
%\marginpar{if lack of space, this diagram could be omitted, it's not so informative}
As a first contribution, in Section \ref{section:4:proof:monad} we give a direct proof of the fact that $\lcset$ is indeed a monad on $\Met$. 
%\marginpar{Qui ho cambiato terminologia. Mi e' venuto il terrore che sta cosa segua dalla C(X) di Tix e Plotkin su domains.}
This result does not seem straightforward to prove. Most notably, establishing the non--expansiveness of the monad multiplication $\mu^{\lcset}$ requires some detailed calculations. 

Our second and main result concerns the presentation of the $\Met$ monad $\lcset$. Presentations of monads in $\Sets$ are given in terms of categories of algebras (in the sense of universal algebra) and their homomorphisms, but these are not adequate in the metric setting. For this reason we use, instead, the recently introduced apparatus of \emph{quantitative algebras} and  \emph{quantitative equational theories} of \cite{radu2016} (see also \cite{DBLP:conf/lics/MardarePP17, DBLP:conf/lics/BacciMPP18, BacciBLM18, DBLP:journals/entcs/Bacci0LM18}). 
%\marginpar{here we could avoid (see also) since it's done just below}
This framework generalises that of universal algebra and equational reasoning by dealing with quantitative algebras, which are metric spaces equipped with non--expansive operations over a signature, and quantitative equations of the form $s=_\epsilon t$, intuitively expressing that the distance between terms $s$ and $t$ is less than or equal to $\epsilon$. In Section \ref{section:4:proof:monad} we define the quantitative equational theory $\qetcs$ of quantitative convex semilattices, and in Section \ref{section:5:proof:presentation} we prove the presentation result (Theorem \ref{thm:main}): the category $\EM(\lcset)$ of Eilenberg-Moore algebras for $\lcset$ is isomorphic to the category $\qacat(\qetcs)$ of quantitative convex semilattices and their non--expansive homomorphims.

\textbf{Relation with other works.}  This work continues the research path opened in the seminal \cite{radu2016} (see also subsequent works  \cite{DBLP:conf/lics/MardarePP17, DBLP:conf/lics/BacciMPP18, BacciBLM18, DBLP:journals/entcs/Bacci0LM18}) where the authors investigated the connection between the quantitative theories of semilattices  ($\qetsl$) and convex algebras ($\qetca$) and the monads $\lfpset$ and $\ldset$, which are $\Met$ variants of $\fpset$ and $\dset$, respectively.  Hence, our work  constitutes a natural step forward. From a technical standpoint, there is a difference between our main presentation result  and those of \cite{radu2016} regarding $\qetsl$ and $\qetca$ (corollaries 9.4 and 10.6 respectively in \cite{radu2016}). Indeed, in \cite{radu2016} the authors only provide representations of the free objects in the categories $\qacat(\qetsl)$ and  $\qacat(\qetca)$. While this suffices in many applications, we believe that proving a full presentation, in the sense introduced and investigated in this work, provides a more general and useful result, giving a representation for the whole categorical structure and not just for free objects. This said, the technical machinery developed in \cite{radu2016} suffices, with minor additional work\footnote{\label{footnote:introduction}The proof structure of our Theorem \ref{thm:main} can be adapted (and in fact much simplified due to the simpler nature of $\qetsl$ and $\qetca$ compared to $\qetcs$) to obtain these isomorphisms of categories.}, to establish the following presentation results in our sense: $\qacat(\qetsl) \cong \EM(\lfpset)$ and $\qacat(\qetca) \cong \EM(\ldset)$.
%\marginpar{little corrections, especially $\qetsl$ instead of $\qetcs$}

\textbf{Note:} Full proofs of the results presented in this paper are available in the Appendix. 
%\textbf{Note:} Due to space limitations, not all results stated in this work have corresponding proofs in the main body of the paper, but full proofs are are available in the Appendix. 

\section{Monads on Sets and Equational Theories}\label{sec:monads:set}
%!TEX root = paper.tex

\label{section:2}

In this section we present basic definitions and results regarding monads. We assume the reader is familiar with the basic concepts of category theory (see \cite{Awodey} as a reference).

\begin{definition}\label{monad:main_definition}
Given a category $\Cat$, a monad  on $\Cat$ is a triple $(\mon, \eta, \mu)$ composed of a functor $\mon\colon\Cat \rightarrow \Cat$ together with two
natural transformations: a unit $\eta\colon id
\Rightarrow \mon$, where $id$ is the identity functor on $\Cat$, and a multiplication $\mu \colon \mon^{2} \Rightarrow
\mon$, satisfying the two laws 
$\mu \circ \eta\mon = \mu \circ \mon\eta = id $ and $  \mu\circ \mon\mu = \mu \circ\mu\mon.
$
\end{definition}
%\marginpar{it's monad on Met, not monad in Met}
%\subsection{Examples of Monads in $\Sets$}

We now introduce three relevant monads on the category $\Sets$ of sets and functions.
%\marginparblue{changed this sentence}
%that have found a wide applications in the formal semantics of programs.  
%Within the following definitions we also establish some notation. 

%\marginparblue{
%%%this should be said/unified maybe: We sometimes write $\mon$ to denote the monad $(\mon, \eta, \mu)$. For $\calg,\calg$ functors over the same category, we often write $\calf\calg(X)$ for $\calf(\calg(X))$, i.e., $\calf\circ \calg(X)$. 
%For an arrow $f$ in the category, we write $\calf f$ for $\calf (f)$? to check}

%(e.g., regarding probability distributions, convex combinations, etc.)

\begin{definition}\label{def:set:fpset}
The \emph{non--empty finite powerset} monad $(\fpset, \eta^{\fpset}, \mu^{\fpset})$ on $\Sets$ is defined as follows.
Given an object $X$ in $\Sets$, $\fpset(X) = \{ X^\prime \subseteq X \mid X^\prime\neq \emptyset \textnormal{ and $X^\prime$ is finite} \}$. Given an arrow $f:X\rightarrow Y$, $\fpset{(f)}: \fpset(X) \to \fpset(Y)$ is defined as $\fpset{(f)}(X^\prime)=\bigcup_{x\in X^\prime} f(x)$ for any $X'\in \fpset(X)$.
The unit $\eta^\fpset_X:X\rightarrow \fpset(X)$ is defined as $\eta^\fpset_X(x) = \{x\}$,
and the multiplication $\mu^\fpset_{X}:  \fpset\fpset(X) \rightarrow \fpset(X)$ is defined as $\mu^\fpset_{X}(\{ X_1,\dots, X_n\})=\bigcup^n_{i=1}X_i$.
\end{definition}

\begin{comment}

\begin{definition}\label{def:set:upset}
The \emph{unrestricted non--empty  powerset} monad $\upset$ on $\Sets$ is defined as follows:
\begin{itemize}
\item Given an object $X$ in $\Sets$, $\upset(X) = \{ X^\prime \subseteq X \mid X^\prime\neq \emptyset \}$. Given an arrow $f:X\rightarrow Y$, $\upset{f}= \big(X^\prime\mapsto \bigcup_{x\in X^\prime} f(x) \big)$,
\item $\eta^\upset_X:X\rightarrow \fpset(X)$ is defined as $\eta^\upset_X(x) = \{x\}$,
\item $\mu^\upset_X:  \upset(\upset(X)) \rightarrow \upset(X)$ is defined as $\{ X_i\}_{i\in I} \mapsto \bigcup_{i}X_i$
\end{itemize}
\end{definition}

\end{comment}

A probability distribution on a set $X$ is a function $\Delta: X\rightarrow [0,1]$ such that $\sum_{x\in X} \Delta(x)=1$. 
The \emph{support} of  $\Delta$ is defined as the set $\support(\distr) = \{ x\in X \mid \distr(x) \neq 0\}$. In this paper we only consider probability distributions with finite support (f.s.), which we sometimes just call {distributions}. The Dirac distribution $\dirac x$ is defined as $\dirac x (x^\prime)=1$ if $x^\prime = x$ and $\dirac x (x^\prime)=0$ otherwise. We often denote a distribution having $\support(\distr) = \{x_{1},x_{2}\}$  using the expression $p_{1} x_{1} + p_{2} x_{2}$, with $p_i = \distr(x_i)$.
Analogously, we let $\sum_{i=1}^{n} p_{i} x_{i}$ denote a distribution $\distr$ with support $\{x_1,\dots, x_n \}$ and with $p_i = \distr(x_i)$.

\begin{definition}\label{def:set:dset}
The \emph{finitely supported probability distribution} monad $(\dset, \eta^{\dset},\mu^{\dset})$ on $\Sets$ is defined as follows.
For objects $X$ in $\Sets$, $\dset(X) = \{ \distr \mid \distr \textnormal{ is a f.s. distribution on $X$}\}$.
%\marginparblue{f.s.}
For arrows $f\!:\!X\rightarrow\! Y$ in $\Sets$, $\dset{(f)}\!:\!\dset{(X)}\!\rightarrow\!\dset{(Y)}$ is defined as
$\dset{(f)}(\distr)\! =\!   \big(y\mapsto \sum_{x \in f^{-1}(y)} \distr(x) \big)$. 
%\marginparblue{tolte graffe nell'indice della sommatoria}
%\marginpar{keep or remove this notation for distributions using $\mapsto$? don't know if it should be explained.}
%$\dset{(f)}(\distr) =   y\mapsto \sum\big\{ \distr(x) \mid  y=f(x)\big\}$.
%The finitely supported probability distribution $\dset{f}(d)\in\dset(T)$ is called the \emph{push--forward distribution}.
The unit $\eta^\dset_X:X\rightarrow \dset(X)$ is defined as $\eta_X(x) = \dirac x$. The multiplication $\mu^{\dset}_X:  \dset\dset(X) \rightarrow \dset(X)$ is defined, for $\sum_{i=1}^{n} p_{i} \Delta_{i} \in \dset\dset(X)$, as $\mu^{\dset}_X(\sum_{i=1}^{n} p_{i} \Delta_{i})= \big( x\mapsto \sum_{i=1}^n p_i\cdot \distr_i(x) \big).$
\end{definition}

\begin{remark}
Given elements $\distr_1,\dots, \distr_n\in \dset(X)$, the expression
$\sum_{i=1}^{n} p_{i} \Delta_{i}$ denotes an element in $\dset\dset(X)$.
The set $\dset(X)$ can be seen as a convex subset of the real vector space $\mathbb{R}^X$, so in order to avoid confusion with the notation 
$\sum_{i=1}^{n} p_{i} \Delta_{i}$
% for expressing distributions of distributions, 
we will use the following dot--notation $\sum_{i=1}^{n} p_{i} \cdot \Delta_{i}$  to denote convex combinations of distributions:
$\sum_{i=1}^{n} p_{i} \cdot \Delta_{i} = \mu^\dset_X ( \sum_{i=1}^{n} p_{i} \Delta_{i} )  = \big(x\mapsto \sum_{i=1}^n p_i\cdot \distr_i(x)\big).$
Hence, %in accordance with the notation established above, 
$\sum_{i=1}^{n} p_{i} \Delta_{i}$ denotes an element in $\dset\dset(X)$ (a distribution of distributions), while $\sum_{i=1}^{n} p_{i} \cdot \Delta_{i}$ denotes an element of $\dset(X)$.
%(the distribution obtained as a convex combination of $\distr_i$).
\end{remark}

\vvcut
{Lastly, we now introduce the convex powerset monad $\cset$ on $\Sets$ which is used, in the field of program semantics, to model nondeterministic and probabilistic effects. In order to do this, recall that the set $\dset(X)$ can be seen as a convex subset of the real vector space $\mathbb{R}^X$. In order to avoid confusion with the notation introduced above (e.g., $p_1x_1 + \dots + p_nx_n$) for expressing probability distributions, we will use the following dot--notation $p_1\cdot \distr_1+ \dots + p_n\cdot \distr_n$  to denote convex combinations of probability distributions:
$p_1\cdot \distr_1+ \dots + p_n\cdot \distr_n = \mu^\dset_X \big( p_1\distr_1+ \dots + p_n \distr_n  \big)  = \big(x\mapsto \sum_{i=1}^n p_i\cdot \distr_i(x)\big).$
%\begin{remark}
Note that, %in accordance with the notation established above, 
given elements $\distr_1,\dots, \distr_n\in \dset(X)$, the expression
$p_1\distr_1 + \dots + p_n\distr_n$ denotes an element in $\dset(\dset(X))$ (a distribution of distributions), while $ p_1\cdot \distr_1 + \dots + p_n\cdot \distr_n
$ denotes an element of $\dset(X)$ (the distribution obtained as a convex combination of $\distr_i$).
%are different. The one on the left denotes an element in $\dset(\dset(X))$ (a distribution of distributions) while the one on the right denotes an element of $\dset(X)$ (the distribution obtained as a convex combination of $\distr_i$). 
%\end{remark}
}
Given a collection $S\subseteq\dset (X)$ of distributions, we can construct its \emph{convex closure} $\conv(S) = \{ \sum_{i=1}^{n} p_{i} \cdot \Delta_{i} \mid  n\geq 1,  \distr_i \in S \textnormal{ for all $i$, and } \sum_{i=1}^{n} p_i = 1 \}.
$ 
%$$
%{\conv(S) = \{  \distr\in \dset X \mid \distr = p_1\cdot \distr_1 + \dots + p_n\cdot \distr_n, \ \distr_i \in S\textnormal{ and } \sum_i p_i = 1 \}}
%$$
Note that $\conv(\conv(S)) = \conv(S)$. A subset $S\subseteq\dset (X)$ is \emph{convex} if $S = \conv(S)$. We say that a convex set $S \subseteq\dset (X)$ is \emph{finitely generated} if there exists a finite set $S^\prime\subseteq \dset (X)$ (i.e., $S^\prime\in\fpset\dset{(X)}$) such that $S=\conv(S^\prime)$.
%Elements of $\cset (X)$ are generated by one minimal, unique set of distributions.
Given a finitely generated convex set $S\subseteq \dset (X)$, there exists one minimal (with respect to the inclusion order) finite set $\ub(S)\in\fpset\dset X$  such that $S=\conv(\ub(S))$. The finite set $\ub(S)$ is referred to as the \emph{unique base} of $S$ (see, e.g., \cite{BSV20ar}). The distributions in $\ub(S)$ are convex--linear independent, i.e.,  if $\ub(S) = \{\distr_1,\dots, \distr_n\}$, then for all $i$, $\distr_i \notin \conv(\{ \distr_j \, |\, j\neq i \})$.
%We are now ready to introduce the monad $\cset$ on $\Sets$.

\begin{definition}\label{def:set:cset}
The \emph{finitely generated non-empty convex powerset of distributions} monad $(\cset, \eta^{\cset}, \mu^{\cset})$ on $\Sets$ is defined as follows. Given an object $X$ in $\Sets$, $\cset(X)$ is the collection of non-empty finitely generated convex sets of finitely supported probability distributions on $X$, i.e., $\cset(X) = \{  \conv(S) \mid S\in \fpset\dset X \}$.
Given an arrow $f:X\rightarrow Y$ in $\Sets$, the arrow $\cset{(f)}: \cset(X)\rightarrow \cset(Y)$ is defined as $\cset{(f)}(S)= \{ \dset{(f)}(\distr) \mid \distr\in S \} $.
%\vvcut
%{The fact that the so defined $\cset{f}$ indeed maps finitely generated convex sets to finitely generated convex sets is easy to verify and well--known (see, e.g., \cite{???}).
%\todo{I would avoid these kind of comments(also below), which already come from the references to $\cset$ being a monad}}
The unit $\eta_X^{\cset}:X\rightarrow \cset(X)$ is defined as $\eta^\cset_X (x) = \{\dirac x\}$, the singleton (convex) set consisting of the Dirac distribution.
The mutiplication $\mu^\cset_X:  \cset\cset(X) \rightarrow \cset(X)$ is defined, for any $S\in \cset\cset(X)$, as
$
\mu^\cset_X(S)= \bigcup_{\Delta \in S}  \wms(\distr)$, 
where, for any $\distr\in \dset\cset (X)$ of the form $\sum_{i=1}^{n} p_{i} S_{i}$, with $S_i\in\cset (X)$, the \emph{weighted Minkowski sum} operation $\wms: \dset\cset (X)\rightarrow \cset (X)$  is defined as 
$
\wms(\distr) = \{ \sum_{i=1}^{n} p_{i} \cdot \Delta_{i} \mid \textnormal{for each $1\leq i \leq n$, $\distr_i\in S_i$}\}.
$
\vvcut{The fact that  $ \big\{ \wms(\distr) \mid  \distr \in S  \big\}   $ is indeed convex and finitely generated whenever the input $S\in  \cset\cset X$ is convex and finitely generated is well--known (see, e.g., \cite{}).}
\end{definition}

\subsection{Equational Theories and Monad Presentations}

An important concept regarding monads is that of algebras for a monad.
\begin{definition}\label{def:algebra-of-a-monad}
Let $(\mon:\Cat \rightarrow \Cat,\eta,\mu)$ be a monad. An algebra for $\mon$ is a pair $(A,h)$ where $A\in\Cat$ is an object and $h:\mon (A)\rightarrow A$ is a morphism such that: $ h \circ  \eta_A = id_A$ and $h \circ \mon h= h \circ \mu_A $.
%\marginparblue{check equations}
Given two $\mon$--algebras $(A,h)$ and $(A^\prime,h^\prime)$, a \emph{$\mon$--algebra morphism} is an arrow $f:A\rightarrow A^\prime$ in $\Cat$ such that
$
 f\circ h = h^\prime \circ \mon(f)  
$.
The category of \ema algebras for $\mon$, denoted by $\EM(\mon)$, has $\mon$--algebras as objects and $\mon$--morphisms as arrows.
\end{definition}

The definitions above are purely categorical and, as a consequence, the category $\EM(\mon)$ is sometimes hard to work with as an abstract entity. It is therefore very useful when $\EM(\mon)$ can be proven isomorphic to a
category whose objects and morphisms are well--known and understood. This leads to the concept of \emph{presentation of a monad}. Before introducing it, we recall some basic definitions of universal algebra (see \cite{univalgebrabook} for a standard introduction).

\begin{definition}\label{basic:definitions:universal-algebra}
A signature $\Sigma$ is a set of function symbols each having its own arity.
% (operations of arity $0$ are constants). 
We denote with $\terms X \Sigma$ the set of terms built from a set of generators $X$ with the function symbols of $\Sigma$. An equational theory $\et$ of type $\Sigma$ is a set  $\et \subseteq {\terms X \Sigma} \times  {\terms X \Sigma } $ of equations between terms  $\terms X \Sigma$ closed under deducibility in the logical apparatus of equational logic. Given a set $E\subseteq  {\terms X \Sigma} \times  {\terms X \Sigma}$ of equations, the theory induced by $E$ is the smallest equational theory containing $E$. The models of a theory $\et$ are  \emph{$\Sigma$--algebras of the theory $\et$},  i.e., structures $(A,\{ f^A\}_{f\in\Sigma})$ consisting of a set $A$ and operations $f^A:A^{ar(f)}\rightarrow A$, for each operation symbol $f\in \Sigma$ having arity $ar(f)$, satisfying all (universally quantified) equations in $\et$. A homomorphism from  $(A,\{ f^A\}_{f\in\Sigma})$  to  $(B,\{ f^B\}_{f\in\Sigma})$ is a function $g:A\rightarrow B$ such that $g(f^A(a_1,\dots, a_n)) = f^B( g(a_1), \dots, g(a_n))$, for all $f\in \Sigma$.  We denote with $\acat(\et)$  the category whose objects are models of the theory $\et$ and morphisms are homomorphisms.%, i.e., structure preserving maps. 
\end{definition}
%\marginparblue{rileggere. Inoltre, correggere in tutto il paper terms sigma X vs. terms X sigma}
%\marginparblue{check functional, which should be function}

%We recall the following well--known useful result.

%\begin{lemma}\label{lemma:free-algebra-folklore}
%Let $\Sigma$ be a signature, $\et$  an equational theory of type $\Sigma$ and ${\terms X \Sigma}_{/\et}$ be the set of $\et$--equivalence classes of $\Sigma$--terms built from $X$. Operations $f\in\Sigma$ can be defined on ${\terms X \Sigma}_{/\et}$ as follows:
%$$
%f([s_1]_{/\et},\dots, [s_n]_{/\et}) = [ f(s_1,\dots, s_n)  ]_{/\et}
%$$ 
%and the definition does not depends on any specific choice of representatives of the equivalence classes. The resulting algebra of type $\Sigma$ is a model of $\et$ and it is, up--to isomorphism, the free $\et$--algebra generated by $X$.
%\end{lemma}

\begin{definition}[Presentation of $\Sets$ monads] Let $\mon$ be a monad on $\Sets$. A presentation of $\mon$ is an equational theory $\et$ such that the categories $\EM(\mon)$ and $\acat(\et)$ are isomorphic.
\end{definition}
%\finaltodo{Check references in thm 13}

In what follows we introduce equational theories that are presentations of the three $\Sets$ monads $\fpset$, $\dset$ and $\cset$ introduced earlier.

\begin{definition}\label{def:semilattices:set}
The theory $\etsl$ of semilattices is the theory having as signature $\sigsl=\{\oplus\}$ and equations stating that $\oplus$ is associative, commutative, and idempotent:\\
\begin{tabular}{l}
(A)\;
$(x\oplus y)\oplus z = x\oplus (y\oplus z)$
\qquad
(C)\; 
$x\oplus y = y\oplus x$
\qquad 
(I)\;
$x\oplus x=x.$
\end{tabular}
%$(x\oplus y)\oplus z \stackrel{(A)}{=} x\oplus (y\oplus z)$, commutative $x\oplus y \stackrel{(C)}{=} y\oplus x$, and idempotent $x\oplus x\stackrel{(I)}{=}x$.
%\[(x\oplus y)\oplus z \stackrel{(A)}{=} x\oplus (y\oplus z)\qquad x\oplus y \stackrel{(C)}{=} y\oplus x\qquad x\oplus x\stackrel{(I)}{=}x\]
%$$\begin{array}{ccc}
%(x\oplus y)\oplus z& \stackrel{(A)}{=}& x\oplus (y\oplus z) \\
%x\oplus y& \stackrel{(C)}{=}& y\oplus x\\
%x\oplus x&\stackrel{(I)}{=}&x
%\end{array}$$
\end{definition}

\begin{definition}\label{def:convex:set}
The theory $\etca$ of convex algebras has signature $\sigca=\{+_p\}_{p\in(0,1)}$ and, for all $p,q\in(0,1)$, the equations for probabilistic associativity, commutativity, and idempotency:\\
%which are the quantitative correspondents of associativity $(x+_qy)+_pz \stackrel{(A_p)}{=} x+_{pq}(y+_{\frac{p(1-q)}{1-pq}}z)$, commutativity $x+_py \stackrel{(C_p)}{=}  y+_{1-p}x$, and idempotency $x+_px\stackrel{(I_p)}{=}x$.
\begin{tabular}{l}
(A$_{p}$)\;
$(x+_qy)+_pz = x+_{pq}(y+_{\frac{p(1-q)}{1-pq}}z)$
\qquad
(C$_{p}$)\; 
$x+_py =  y+_{1-p}x$
\qquad 
(I$_{p}$)\;
$x+_px=x.$
\end{tabular}
%\[(x+_qy)+_pz \stackrel{(A_p)}{=} x+_{pq}(y+_{\frac{p(1-q)}{1-pq}}z)\qquad
%x+_py \stackrel{(C_p)}{=}  y+_{1-p}x\qquad
%x+_px\stackrel{(I_p)}{=}x\]
%$$\begin{array}{ccc}
%(x+_qy)+_pz& \stackrel{(A_p)}{=}& x+_{pq}(y+_{\frac{p(1-q)}{1-pq}}z)\\
%x+_py& \stackrel{(C_p)}{=} & y+_{1-p}x\\
%x+_px&\stackrel{(I_p)}{=}&x
%\end{array}$$
\end{definition}

\begin{definition}\label{def:convexsemilattices:set}
The theory $\etcs$ of convex semilattices is the theory with signature $\sigcs=(\{\oplus\}\cup \{+_p\}_{p\in(0,1)})$ where $\oplus$ satisfies the equations of semilattices, $\pplus p$ satisfies the equations of convex algebras for every $p\in (0,1)$, and, furthermore, for every $p\in (0,1)$ the following distributivity equation (D) is satisfied: $x+_p (y \oplus z) = (x+_p y)\oplus (x+_p z)$. 
%The category of convex algebras with their homomorphisms is denoted by $\mathbf{CS}$.
\end{definition}
The following proposition collects known results in the literature (see \cite{swirszcz:1974, Doberkat0608, jacobs:2010, BSV19}).\begin{proposition}$ $
\vspace{-0.25cm}
\begin{enumerate}
\itemsep-0.2cm
\item The theory $\etsl$ of semilattices is a presentation of $\pset$, i.e., $\acat(\etsl)\cong \EM(\pset)$.
\item The theory $\etca$ of convex algebras is a presentation of $\dset$, i.e., $\acat(\etca)\cong \EM(\dset)$.
\item The theory $\etcs$ of convex semilattices is a presentation of $\cset$, i.e., $\acat(\etcs)\cong \EM(\cset)$.
\end{enumerate}
\end{proposition}

%%%%%%%%%%%%%%%%%%%%%%%%

\subsubsection{One Application: Representation of Term Algebras}

Having presentations of $\Sets$ monads as categories for algebras of equational theories is mathematically convenient for several reasons. One useful application, especially in the field of program semantics, are representation theorems for free algebras, i.e., term algebras.

In this section we assume the reader to be familiar with the concept of free object in a category (see, e.g., \cite[\S 10.3]{Awodey}).
%\vvcut
%{\begin{definition}
%Let $\Cat$ be a category and let $F:\Cat \rightarrow \Sets$ be a functor. Let $X$ be a set (of generators) in $\Sets$, let $A$ be an object in $\Cat$ and let $i:X\rightarrow F(A)$ be an injective map between sets. We say that $A$ is the \emph{free object on $X$ (with respect to $i$)} if and only if for any object $B\in\Cat$ and any map between sets $f:X\rightarrow F(B)$, there exists a unique morphism $g:A\rightarrow B$ such that $f=F(g)\circ i$. 
%\todo{add diagram}.
%\end{definition}
%Free objects in $\acat(\et)$, the category of algebras for a given equational theory $\et$ of a given type $\Sigma$, and in $\EM(\mon)$, the category of \ema algebras for a $\Sets$--monad $\mon$, are well-known. }
The free object generated by $X$ in the category $\EM(\mon)$ is the $\mon$--algebra $(\mon(X), \mu^{\mon}_{X})$. The free object generated by $X$ in the category $\acat(\et)$ is the term algebra, i.e., the algebra whose the carrier is ${\terms X \Sigma}_{/\et}$, the set of $\Sigma$--terms constructed from the set of generators $X$ taken modulo the equations of the theory $\et$, and with operations defined on equivalences classes, that is : $f([t_1]_{/\et},\dots, [t_n]_{/\et}) = [f(t_1,\dots, t_n)]_{/\et}$ for each $f\in\Sigma$. 
%\marginpar{VVV: I put generated by X instead of on X, takes no extra line and it's more consistent with what we say for metrics}
\vvcut{The definition does not depend on any specific choice of representatives for the equivalence classes.}
%\begin{proposition}
%Let $\mon:\Sets\rightarrow\Sets$ be a monad and let $X$ be a set. The free object on $X$ in the category $\EM(\mon)$ is the algebra $(\mon(X), \mu^{\mon}_{X})$.
%\end{proposition}
%\begin{proposition}\label{term:algebra:set:initial:theorem}
%Let $\Sigma$ be a signature, $X$ a set and $\et\subseteq {\terms X \Sigma} \times {\terms X \Sigma}$ an equational theory of type $\Sigma$. The free object on $X$ in the category $\acat(\et)$ is the term algebra, i.e., the algebra whose the carrier is ${\terms \Sigma X }_{/\et}$, the set of $\Sigma$--terms constructed from the set of generators $X$ taken modulo the equations of the theory $\et$, and with operations defined on equivalences classes, that is : $f([t_1]_{/\et},\dots, [t_n]_{/\et}) = [f(t_1,\dots, t_n)]_{/\et}$ for each $f\in\Sigma$.
%The definition does not depend on any specific choice of representatives for the equivalence classes.
%\end{proposition}
These characterizations, together with the fact that free objects are unique up to isomorphism, can be used to derive the following result. 
\begin{proposition}\label{prop:free-algebras}
Let $\mon$ be a monad on $\Sets$ and let $F:\acat(\et)\cong \EM(\mon)$ be a presentation of $\mon$ in terms of the equational theory $\et$ of type $\Sigma$. Then the term algebra ${\terms  X \Sigma}_{/\et}$ and the free \ema algebra $(\mon(X), \mu^{\mon}_{X})$ are isomorphic (via $F$).
\end{proposition}
%\marginpar{VV: qui scrivevi free term algebra, ma mi sembra ridondante. Inoltre, questo ora e' una proposition}
In other words, a presentation theorem for $\mon$ provides automatically representation results for term algebras %, %which are thus given a clear characterization 
via the known semantic behaviour of the multiplication of $\mon$. 

\begin{example}\label{example:presentation:set:3}
The presentation of the monad $\cset$ in terms of the theory of convex semilattices implies that the free convex semilattice generated by $X$ is isomorphic with the convex semilattice $(\cset X ,\oplus,+_p)$ where $S_1\oplus S_2 = \conv(S_1\cup S_2)$ (convex union) and $S_1 +_p S_2 = \wms ( p S_1 + (1-p)S_2)$ (weighted Minkowski sum), for all $S_1,S_2\in\cset (X)$. In other words, the set ${\terms  X {\sigcs}}_{/\etcs}$ of  convex semilattice terms modulo the equational theory of convex semilattices can be identified with the set $\cset (X)$ of finitely generated convex sets of finitely supported probability distributions on $X$.
The isomorphism is explicitly given in \cite{BSV20ar} by the function 
$\ics : \cset(X) \to {\terms  X {\sigcs}}_{/\etcs} $
defined as 
$\ics(S) = [\bigcplus_{\distr\in \ub(S) } (\bigpplus_{x \in \support(\distr) } \distr(x)\,x)]_{/\etcs}$, where $\bigcplus_{i\in I} x_{i}$ and $\bigpplus_{i\in I} p_{i} \,x$ are respectively notations for the binary operations $\cplus$ and $\pplus p$ extended to operations of arity $I$, for $I$ finite (see, e.g., \cite{stone:1949, BSS17}).
%\footnote{Note that a chosen ordering of elements in $X$ and of distributions in $\dset(X)$ is assumed in the definition. The axioms of convex semilattices guarantee that the specific choice made is not relevant.}
%\marginparblue{footnote puo'essere eliminato}
It is useful to stress the important role played by the axiom (D) of  convex semilattice. This equation allows us to derive the convexity equation $x \oplus y = x \oplus y \oplus (x \pplus p y)$ (see, e.g., \cite[Lemma 14]{BSV20ar}).
%\marginparred{Ho spostato questa frase dentro l'esempio perche' mi sembra confinata li VVV: okay ma sono andata a capo perche' e' un'altra storia comunque}
\end{example}

We remark that presentation results of monads may provide more insights than just the application given by Proposition \ref{prop:free-algebras}, as they give representation for the whole categorical structure and not just free objects.

\section{Monads on Met and Quantitative Equational Theories}
\label{sec:monadmet}
%!TEX root = paper.tex

In Section \ref{sec:monads:set} we have considered monads in the category $\Sets$. We now shift our focus to monads in the category $\Met$ of metric spaces and non--expansive functions. The category $\Met$ provides a natural mathematical setting for developing the semantics of programs exhibiting quantitative behaviour such as, e.g., probabilistic choice. It is indeed appropriate in this setting to replace the usual notion of program equivalence with the more informative notion of program distance (see, e.g., \cite{prakashbook,GJS90,breugel2005,DJGP02,DBLP:conf/icalp/BreugelW01}). 

\vvcut{We first recall some basic definitions regarding metric spaces.}

\begin{definition}
A metric space is a pair $(X,d)$ such that $X$ is a set and $d:X\rightarrow\mathbb{R}$ is a function, called the \emph{metric}, satisfying the following properties: $d(x,x) = 0$, $d(x,y)=d(y,x)$ and $d(x,y) \leq d(x,z) + d(z,y)$, for all $x,y,z\in X$. A function $f:X_1 \rightarrow X_2$ between two metric spaces $(X_1,d_1)$ and $(X_2,d_2)$ is called non--expansive (a.k.a.\ $1$--Lipschitz) if $d_2( f(x), f(y) ) \leq d_1(x,y)$ for all $x,y\in X_1$. We denote with $\Met$ the category whose objects are metric spaces and whose morphisms are non--expansive maps.
\end{definition}

Given two metrics $d_1,d_2$ on $X$, we write $d_1\funleq d_2$ if for all $x,x^\prime\in X$, it holds that $d_1(x,x^\prime)\leq d_2(x,x^\prime)$. Let $(Y,d)$ be a metric space,  $X$ a set and $f:X\rightarrow Y$. We write $d\funpair f f$ for the metric on $X$ defined as $d\funpair f f  ( x_1,x_2) =   d(f(x_1), f(x_2))$.  Let $d_{\mathbb{R}}$ be the Euclidean metric  on  $\mathbb{R}$  defined as $d_\mathbb{R}(r_1,r_2) = | r_1 - r_2|$. If $(X,d)$ is a metric space, we simply say that $f:X\rightarrow[0,1]$ is non--expansive to mean that $f:(X,d)\rightarrow ([0,1],d_{\mathbb{R}})$ is non--expansive. The metric $d$ of a metric space $(X,d)$ induces a topology on $X$ whose open sets are generated by the \emph{open balls} of the form $B(x,\epsilon)=\{ y\in X \mid d(x,y) < \epsilon\}$, for $x\in X$ and $\epsilon >0$. A subset $Y\subseteq X$ is called \emph{compact} if it is closed and bounded (i.e., the distance between elements in $Y$ is bounded by some real number). The collection of non--empty compact subsets of a metric space $(X,d)$ is denoted by $\kompactd X d$. Note that every finite subset of $X$ belongs to $\kompactd X d$.

%\subsection{Examples of Monads in $\Met$}

%In this subsection we introduce the $\Met$ variants of t
The $\Sets$ monads $\fpset$ and $\dset$ defined in Section \ref{sec:monads:set} can be extended to monads in $\Met$. These extensions are well--known and are based on metric liftings constructions due to  Hausdorff and Kantorovich (see \cite{Kechris} for a standard reference).

\begin{definition}[Hausdorff Lifting]
Let $(X,d)$ be a metric space. The Hausdorff lifting of $d$ is a metric $\haus(d)$ on $\kompactd X d$,  the collection of non--empty compact subsets of $X$, defined as follows for any pair $X_1,X_2 \in \kompactd X d$:
$$
\haus(d)\big(X_1,X_2) = \max\big\{  \sup_{x_1\in X_1}\inf_{x_2\in X_2}d(x_1,x_2)  \ \     , \ \ \sup_{x_2\in X_2}\inf_{x_1\in X_1}d(x_1,x_2)    \big\}.
$$
\end{definition}
%Clearly, given any collection $\mathcal{C}\subseteq\kompact(X)$ of compact subsets of $X$, the metric $\haus(d)$ restricted to $\mathcal{C}$ is a metric on $\mathcal{C}$. For example, taking  $\mathcal{C}= \fpset{X}$, we have that $(\fpset(X), \haus(d))$ is a metric space.
This leads to the well--known \emph{hyperspace} monad $\hs$ on $\Met$ (\cite{Hausdorff14}, see also \cite{Kechris}).\footnote{%This monad is essentially due to Hausdroff \cite{Hausdorff14}. 
Variants of this monad can be defined in other categories, such as the Vietoris monad 
on compact Hausdorff spaces and continuous functions (see, e.g., \cite{Garner20}). 
%One of these variants is well known as the Vietoris monad. see, e.g., \cite{Kechris} for a detailed exposition.
}

\begin{definition}\label{def:mett:compact}
The \emph{hyperspace} monad $(\hs, \eta^\hs, \mu^\hs)$ on $\Met$ is defined as follows.
Given an object $(X,d)$ in $\Met$, $\hs(X,d) = \big( \kompactd X d, \haus(d) \big)$, the metric space of non--empty compact subsets of $X$ equipped with the Hausdorff distance. Given a non--expansive map $f:(X,d_X)\rightarrow (Y,d_Y)$, $\hs (f)(X^\prime)=\bigcup_{x\in X^\prime} f(x)$.
The unit $\eta^\hs_{(X,d)}: (X,d) \rightarrow \hs(X,d)$ is defined as $\eta^\hs_{(X,d)}(x) = \{x\}$, and the multiplication $\mu^\hs_{(X,d)}:  \hs\hs(X,d) \rightarrow \hs(X,d)$ is defined as $\mu^\hs_{(X,d)}(\{ X_i\}_{i\in I})=\bigcup_{i}X_i$.
\end{definition}

%The fact that the above definitions are correct (i.e., that $\hs$ is a functor, that $\eta^\hs$ and $\mu^\hs$ are non--expansive and satisfy the monad laws) is well--known (see, e.g., \cite{Kechris}).

The restriction of the monad $\hs$ to finite (hence compact) subsets leads to the following version of the non--empty finite powerset monad on $\Met$, which we denote with $\lfpset$ to distinguish it from the $\Sets$ monad $\fpset$.

\begin{definition}\label{def:lfpset}
The \emph{non--empty finite powerset} monad $(\lfpset, \eta^{\lfpset}, \mu^{\lfpset})$ on $\Met$ is defined as follows.
Given an object $(X,d)$ in $\Met$, $\lfpset(X,d) = \big(\fpset(X), \haus(d) \big)$, the collection of finite non--empty subsets of $X$ equipped with the Hausdorff distance.
The action of $\lfpset$ on morphisms, the unit $\eta^{\lfpset}$ and the multiplication $\mu^{\lfpset}$ are defined as for the $\Sets$ monad $\fpset$ (or, equivalently, as for the $\hs$ monad on $\Met$ restricted to finite sets).
\end{definition}

Next, we introduce the Kantorovich lifting on finitely supported distributions \cite{Kechris}.
%\footnote{The Kantorovich lifting can more generally be formulated using Borel probability measures.}
%The Kantorovich-Rubinstein duality theorem provides an equivalent definition of the lifting.}
%\marginpar{get rid of footnotes? also I removed dual definition of Kant, check that it is not used in appendix}
\begin{definition}[Kantorovich Lifting]\label{def:kantorovich:lifting}
Let $(X,d)$ be a metric space. The Kantorovich lifting of $d$ is a metric $\kant(d)$ on $\dset(X)$, the collection of finitely supported probability distributions on $X$, defined as follows for any pair $\distr_1, \distr_2\in \dset(X)$:
\vvcut{$$
\kant(d) (\distr_1,\distr_2)=
\sup_{\stackrel{f:X\to [0,1]}{\textnormal{non--expansive}}} \Big\{  |  \big(\sum_{x\in X} f(x)\distr_1(x) \big) - \big(\sum_{x\in X} f(x)\distr_2(x) \big) | \Big\}
$$
where $| \_ |$ denotes the absolute value. Equivalently\footnote{The equivalence is a nontrivial result and follows from the Kantorovich-Rubinstein duality theorem.}, the metric $\kant(d)$ can be defined as:}
$$\kant(d) (\distr_1,\distr_2)=
\inf_{\omega \in Coup(\distr_1,\distr_2)}\Big( \sum_{(x_1,x_2)\in X \times X }\omega(x_1,x_2) \cdot d(x_1,x_2) \Big)
$$
where $Coup(\distr_1,\distr_2)$ is defined as the collection of couplings of $\distr_1$ and $\distr_2$, i.e., the collection of probability distributions on the product space $X\times X$ such that the marginals of $\omega$ are $\distr_1$ and $\distr_2$. Formally, 
$
 Coup(\distr_1,\distr_2) = \{ \omega\in\dset (X\times X) \mid  \dset(\pi_1)(\omega) = \distr_1 \textnormal{ and }  \dset(\pi_2) (\omega) = \distr_2 \}
$
where $\pi_1:X_1\times X_2 \rightarrow X_1$ and $\pi_2:X_1\times X_2 \rightarrow X_2$ are the projection functions. 
\end{definition}

We can now introduce the following version of the finitely supported probability distribution monad on $\Met$, which we denote with $\ldset$ to distinguish it from the $\Sets$ monad $\dset$.

\begin{definition}\label{def:ldset}
The \emph{finitely supported probability distribution} monad $(\ldset, \eta^{\ldset}, \mu^{\ldset})$ on $\Met$ is defined as follows. Given an object $(X,d)$ in $\Met$, $\ldset(X,d) = \big(\dset(X), \kant(d) \big)$, the collection of f.s. probability distributions on $X$ equipped with the Kantorovich distance.
The action of $\ldset$ on morphisms, the unit $\eta^{\ldset}$, and the multiplication $\mu^{\ldset}$ are defined as for the $\Sets$ monad $\dset$.
%\marginparblue{f.s.}
\begin{comment}
as follows:
\begin{itemize}
\item Given a morphism $f:(X,d)\rightarrow (Y,d)$, i.e., a non--expansive map $f:X\rightarrow Y$, we define 
$\ldset{f}= \dset{f}$, 
\item $\eta^{\ldset}_X:X\rightarrow \lfpset(X)$ is defined as $\eta^{\ldset}_X(x) = \dirac x$, where $\dirac x$ is the Dirac distribution,
\item  $\mu^{\ldset}_X:  \ldset(\ldset(X)) \rightarrow \ldset(X)$ is defined as: $$(p_1 \distr_1 + \dots + p_n \distr_n) \mapsto \big( x\mapsto \sum_{i=1}^n p_i\cdot \distr_i(x) \big).$$
\end{itemize}
\end{comment}
\end{definition}

The fact that the above definitions are correct (i.e., that $\ldset$ is a functor, that $\eta^{\ldset}$ and $\mu^{\ldset}$ are non--expansive and satisfy the monad laws) is well--known (see, e.g., \cite{Kechris,breugel2005, BaldanBKK18}). 
%\marginpar{found this we could check out http://www.paoloperrone.org/phdthesis.pdf, also maybe cite papers by Fritz and Perrone}

\subsection{Quantitative Equational Theories and Quantitative Algebras}

\vvcut
{The mathematical field of universal algebra studies \emph{algebras}, which are sets equipped with finitary operations, and their \emph{equational theories}, which are sets of equations closed under the deductive apparatus of equational logic.

\begin{center}
Algebras  $\Longleftrightarrow$ Equational Theories.
\end{center}

In a series of recent papers \cite{radu2016}, the authors have started investigating a new logical framework dealing with  \emph{quantitative algebras} which are algebras in $\Met$ (i.e., metric spaces equipped with finitary non--expansive operations) and \emph{quantitative equational theories}.

\begin{center}
Quantitative algebras  $\Longleftrightarrow$ Quantitative Equational Theories
\end{center}

\ntodo{until here this stuff should go in intro}
}

We provide here the essential definitions and results of the framework developed by Mardare, Panangaden, and Plotkin in \cite{radu2016}.
In what follows, a signature $\Sigma$ is fixed. Recall that $\terms X \Sigma$ denotes the set of terms constructed from $X$ using the function symbols in $\Sigma$. A substitution
%%, ranged over by the letter $\sigma$, 
is a map of type $\sigma: X\to \terms X \Sigma$. As usual, to any interpretation $\iota: X\rightarrow A$ of the variables into a set
%the carrier of a quantitative algebra 
corresponds, by homomorphic extension, a unique map $\iota: \terms X \Sigma \rightarrow A$. 

\begin{definition}[Quantitative Equational Theory]\label{def:qet}
A \emph{quantitative equation} is  an expression of the form $t =_{\epsilon} s$, where $ t, s \in \terms X \Sigma$ and $\epsilon \in [0,1]$. We denote with $E(\Sigma)$ the collection of all quantitative equations. We use the letters  $\Gamma, \Theta$ to range over subsets of $E(\Sigma)$. 
A \emph{quantitative inference} is an element of $2^{E(\Sigma)}\times E(\Sigma)$, i.e., a pair $(\Gamma, t=_\epsilon s)$ where $\Gamma\subseteq E(\Sigma)$ and $t=_\epsilon s$ is a quantitative equation. 
Note that $\Gamma$ needs not be finite.
A \emph{deducibility relation} is a set of quantitative inferences $\vdash\ \subseteq 2^{E(\Sigma)}\times E(\Sigma)$ closed under the following conditions which are stated for arbitrary $s,t,u\in \terms X \Sigma$, $\epsilon, \epsilon^\prime\in [0,1]$, $\Gamma,\Theta\subseteq E(\Sigma)$ and $f\in \Sigma$:\\
(Notation: we use the infix notation $\Gamma \vdash t=_\epsilon s$ to mean that $(\Gamma, t=_\epsilon s ) \in\ \vdash$)\\
\begin{tabular}{l}
(Refl)\;
$\emptyset \vdash t =_0 t$
\qquad 
(Symm) \;
$\{t=_\epsilon s\}\vdash s=_\epsilon t$
\qquad
(Triang) \;
$\{t =_\epsilon u,u =_{\epsilon'} s\} \vdash t =_{\epsilon+\epsilon'} s$
\\
(Max)\;
$\{t=_\epsilon s\}\vdash t=_{\epsilon'} s$, where $\epsilon^\prime > \epsilon$
\qquad
(Arch)
$\{t=_{\epsilon'} s\}_{\epsilon' >\epsilon}\vdash t=_\epsilon s$
\\
(NExp)\;
$\{t_1=_\epsilon s_i\}_{i\in {1\dots ar(f)}}\vdash f(t_1,\dots, t_n)=_{\epsilon} f(s_1,\dots s_n)$
\\
(Subst)\;
if $\Gamma \vdash t=_\epsilon s \text{ then } \{ \sigma(t) =_\epsilon \sigma(s) \mid (t=_\epsilon s) \in \Gamma\}\vdash \sigma(t)=_\epsilon \sigma(s)$, for all $\sigma \in S(\Sigma)$
\\
(Cut)\;
 \;if $\Gamma \vdash \Theta \text{ and } \Theta \vdash t=_\epsilon s \text{ then } \Gamma \vdash t=_\epsilon s$
\\
 (Assum)\;
 if  $ t=_\epsilon s\in\Gamma \text{ then } \Gamma\vdash t=_\epsilon s$, for all $\Gamma,t,s,\epsilon$.
\end{tabular}\\
%\begin{tabular}{l l }
%(Refl)\ &
%$\emptyset \vdash t =_0 t$
%\\
%(Symm) &
%$\{t=_\epsilon s\}\vdash s=_\epsilon t$
%\\
%(Triang) &
%$\{t =_\epsilon u,u =_{\epsilon'} s\} \vdash t =_{\epsilon+\epsilon'} s$
%\\
%(Max)&
%$\{t=_\epsilon s\}\vdash t=_{\epsilon'} s$, where $\epsilon^\prime > \epsilon$
%\\
%(NExp)&
%$\{t_1=_\epsilon s_i\}_{i\in {1\dots ar(f)}}\vdash f(t_1,\dots, t_n)=_{\epsilon} f(s_1,\dots s_n)$\\
%
%(Arch) &
%$\{t=_{\epsilon'} s\}_{\epsilon' >\epsilon}\vdash t=_\epsilon s$
%\\
%(Subst)&
%if $\Gamma \vdash t=_\epsilon s \text{ then } \sigma(\Gamma)\vdash \sigma(t)=_\epsilon \sigma(s)$, for all $\sigma \in S(\Sigma)$
%\\
%(Cut)&
% if $\Gamma \vdash \Theta \text{ and } \Theta \vdash t=_\epsilon s \text{ then } \Gamma \vdash t=_\epsilon s$
%\\
% (Assum)&
% if  $ t=_\epsilon s\in\Gamma \text{ then } \Gamma\vdash t=_\epsilon s$, for all $\Gamma,t,s,\epsilon$.
%\end{tabular}\\
where 
%in (Subst) we define $\sigma(\Gamma)$ as $\sigma(\Gamma)= \{ \sigma(s) =_\epsilon \sigma(t) \mid (s=_\epsilon t) \in \Gamma\}$, and
 in (Cut) the expression $\Gamma\vdash \Theta$ means that for all $(t=_\epsilon s)\in\Theta$ it holds that $\Gamma\vdash t=_\epsilon s$.
Given a set of quantitative inferences $\qt \subseteq 2^{E(\Sigma)}\times E(\Sigma)$, the quantitative equational theory induced by $\qt$ is the smallest deducibility relation which includes $\qt$.
%\marginparblue{se abbiamo bisogno di spazio metti nelle regole direttamente queste righe}
%we use interchangeably the notations $\vdash_\qt$  and $\qet_{\qt}$ to denote the smallest deducibility relation which includes $\qt$. The relation $\qet_{\qt}$ is called the \emph{quantitative equational theory} induced by $\qt$.
\end{definition}

The models of quantitative theories are quantitative algebras, which we now introduce. 
%A quantitative algebra is an ordinary algebra $A$ of type $\Sigma$ whose carrier $A$ is endowed with a metric $(A,d_A)$ and all operations $f:A^{ar(f)}\rightarrow A$ are non--expansive. Accordingly, morphisms between quantitative algebras are ordinary homomorphisms that are non--expansive. 

\begin{definition}[Quantitative Algebra]\label{quantitative:algebra:def}
A \emph{quantitative algebra} of type $\Sigma$ is a structure $\alga=\big(A, \{f^A\}_{f\in \Sigma}, d_{A}\big)$ where  $(A,d_{A})$ is a metric space and, for each $f\in \Sigma$, the function $f^A : A^{ar(f)}\rightarrow A$ is a non--expansive map, with $A^{ar(f)}$ endowed with the $\sup$--metric defined as $d_{\sup}(\{a_i\}_{i\in ar(f)}, \{b_i\}_{i\in ar(f)}) = \max_{i \in ar(f)}( d(a_i, b_i))$.
%\marginparblue{recheck/define sup metric}
  A homomorphism between quantitative algebras $\alga$ and $\algb$ of type $\Sigma$ is a non--expansive function $g:(A,d_A)\rightarrow (B,d_B)$ which preserves all operations in $\Sigma$, i.e., $g(f^A(x_1,\dots,x_n) ) = f^B ( g(x_1), \dots, g(x_n))$, for all $x_i\in A$. %We denote with $\qa$ the category whose objects are quantitative algebras and arrows are morphisms between them.
We say that $\mathbb{A}$ satisfies a quantitative inference $(\{ s_i =_{\epsilon_i} t_i\}_{i\in I}, s=_\epsilon t)$, written 
$
\{ s_i =_{\epsilon_i} t_i\} \models_\mathbb{A} s=_\epsilon t,
$
if  for every interpretation $\iota:X\rightarrow A$ of the variables $X$ into elements of $A$ the following holds: if for all $i\in I$,      $d_A\big(\iota(s_i), \iota(t_i)\big) \leq \epsilon_i$, then $d_A\big(\iota(s), \iota(t)\big) \leq \epsilon$.
We say that  $\mathbb{A}$ is a model of a quantitative theory $\qet$ if $\mathbb{A}$ satisfies every quantitative inference in $\qet$. 
%For a quantitative theory $\qet_{\qt}$,
%the quantitative algebra $\mathbb{A}$ is a model of $\qet_{\qt}$
%if $\mathbb{A}$ satisfies every quantitative inference in $\qet_{\qt}$.
We denote with $\qacat(\qet)$  the category having as objects the quantitative algebras that are models of $\qet$, and as arrows the non--expansive homomorphisms between quantitative algebras of type $\Sigma$.
\end{definition}

Every quantitative algebra of type $\Sigma$ satisfies the quantitative inferences generating the deducibility relation $\vdash$ in Definition \ref{def:qet}.
We refer to  \cite{radu2016} for proofs that all the above definitions are indeed well--defined. Two interesting quantitative theories studied in  \cite{radu2016} are the following.

\begin{definition}[Quantitative Semilattices]\label{def:semilattices:met}
The quantitative theory of \emph{quantitative semilattices}, denoted by  $\qet_{SL}$, has type $\Sigma_{SL}$ (see Definition \ref{def:semilattices:set}) and is
induced by the following quantitative inferences, for all $\epsilon_1,\epsilon_2\in[0,1]$:\\
%the quantitative inferences expressing that $\cplus$ is associative 
%$\emptyset \vdash x \oplus ( y \oplus z) =_0 (x \oplus y) \oplus z$, commutative $\emptyset \vdash x \oplus y =_0 y \oplus x$, and idempotent $\emptyset \vdash x\oplus x =_0 x$, and by
%\[(H)\quad 
%\big\{x_1=_{\epsilon_1} x_1^\prime,     x_2=_{\epsilon_2} x_2^\prime\big\}\vdash x_1 \oplus x_2 =_{\max(\epsilon_1,\epsilon_2)} x_1^\prime\oplus x_2^\prime\]
%\marginpar{inline the axioms, possibly without name}
\begin{tabular}{l}
(A)\; 
$\emptyset \vdash x \oplus ( y \oplus z) =_0 (x \oplus y) \oplus z$  
\qquad 
(C)\; 
$\emptyset \vdash x \oplus y =_0 y \oplus x$
\qquad 
(I)\; 
$\emptyset \vdash x\oplus x =_0 x$\\
(H)\; 
$\big\{x_1=_{\epsilon_1} y_1,     x_2=_{\epsilon_2} y_2\big\}\vdash x_1 \oplus x_2 =_{\max(\epsilon_1,\epsilon_2)} y_{1}\oplus y_{2}.$
\end{tabular}
%
%\[(A)\
%\emptyset \vdash x \oplus ( y \oplus z) =_0 (x \oplus y) \oplus z\quad
%(C)\
%\emptyset \vdash x \oplus y =_0 y \oplus x
%\quad
%(I)\ 
%\emptyset \vdash x\oplus x =_0 x\]
%\[
%(H)\ 
%\big\{x_1=_{\epsilon_1} x_1^\prime,     x_2=_{\epsilon_2} x_2^\prime\big\}\vdash x_1 \oplus x_2 =_{\max(\epsilon_1,\epsilon_2)} x_1^\prime\oplus x_2^\prime\]
%
%Let $\Sigma_{SL}$ be the signature of semilattices (see Definition \ref{def:semilattices:set}) and let $\qet_{SL}$ be the quantitative theory generated by the following set  of quantitative inferences:
%\begin{center}
%\begin{tabular}{l l }
%(A)\ &
%$\emptyset \vdash x \oplus ( y \oplus z) =_0 (x \oplus y) \oplus z$
%\\
%(C)\ &
%$\emptyset \vdash x \oplus y =_0 y \oplus x$
%\\
%(I)\ &
%$\emptyset \vdash x\oplus x =_0 x$
%\\
%(H)\ &
%$\big\{x_1=_{\epsilon_1} x_1^\prime,     x_2=_{\epsilon_2} x_2^\prime\big\}\vdash x_1 \oplus x_2 =_{\max(\epsilon_1,\epsilon_2)} x_1^\prime\oplus x_2^\prime$
%\end{tabular}
%\end{center}
%for all $\epsilon_1,\epsilon_2\in[0,1]$. The quantitative theory $\qet_{SL}$ is called the \emph{theory of quantitative semilattices.}
\end{definition}

%$\big\{x_1=_{\epsilon_1} x_1^\prime,     x_2=_{\epsilon_2} x_2^\prime\big\}\vdash x_1 \oplus x_2 =_{\max(\epsilon_1,\epsilon_2)} x_1^\prime\oplus x_2^\prime$.

\begin{definition}[Quantitative Convex Algebras]\label{def:convex:met}
The quantitative theory of \emph{quantitative convex algebras}, denoted by  $\qet_{CA}$, has type $\Sigma_{CA}$ (see Definition \ref{def:convex:set}) and is
induced by the following quantitative inferences, for all $p,q\in (0,1)$ and $ \epsilon_1,\epsilon_2\in[0,1]$:\\
%the quantitative inferences expressing that $\pplus p$ is associative $\emptyset \vdash (x+_qy)+_pz =_0 x+_{pq}(y+_{\frac{p(1-q)}{1-pq}}z)$, commutative $\emptyset \vdash x+_py  =_0  y+_{1-p}x$, and idempotent $\emptyset \vdash x+_px  =_0  x$, and by
%\[(K)\quad 
%\big\{x_1=_{\epsilon_1} x_1^\prime,     x_2=_{\epsilon_2} x_2^\prime\big\}\vdash x_1 =_p x_2 =_{p\epsilon_1 + (1-p)\epsilon_2} x_1^\prime  +_p x_2^\prime\]
\begin{tabular}{l}
(A$_{p}$)\;
$\emptyset \vdash (x+_qy)+_pz =_0 x+_{pq}(y+_{\frac{p(1-q)}{1-pq}}z)$
\qquad
(C$_{p}$)\; 
$\emptyset \vdash x+_py  =_0  y+_{1-p}x$\\
(I$_{p}$)\;
$\emptyset \vdash x+_px  =_0  x$
\qquad
(K)\;
$\big\{x_1=_{\epsilon_1} y_{1},     x_2=_{\epsilon_2} y_{2}\big\}\vdash x_1 +_p x_2 =_{p\cdot \epsilon_1 + (1-p)\cdot \epsilon_2} y_{1}  +_p y_{2}.$
\end{tabular}
%Let $\Sigma_{CA}$ be the signature of semilattices (see Definition \ref{def:convex:set}) and let $\qet_{CA}$ be the quantitative theory generated by the following quantitative inferences:\\
%\begin{tabular}{l l }
%(A$_{p}$)\ &
%$\emptyset \vdash (x+_qy)+_pz =_0 x+_{pq}(y+_{\frac{p(1-q)}{1-pq}}z)$
%\\
%(C$_{p}$)\ &
%$\emptyset \vdash x+_py  =_0  y+_{1-p}x$
%\\
%(I$_{p}$)\ &
%$\emptyset \vdash x+_px  =_0  x$
%\\
%(K)\ &
%$\big\{x_1=_{\epsilon_1} x_1^\prime,     x_2=_{\epsilon_2} x_2^\prime\big\}\vdash x_1 =_p x_2 =_{p\epsilon_1 + (1-p)\epsilon_2} x_1^\prime  +_p x_2^\prime$
%\end{tabular}\\
%for all $p,q, \epsilon_1,\epsilon_2\in[0,1]$. %The quantitative theory $\qet_{CS}$  is called the \emph{theory of quantitative convex algebras.}
\end{definition}
In other words, the theories $\qet_{SL}$ and $\qet_{CA}$ are obtained by taking the equational axioms of semilattices and convex algebras respectively (Definitions \ref{def:semilattices:set} and \ref{def:convex:set}), replacing the equality $(=)$ with $(=_0)$, and by introducing the quantitative inferences  (H) and (K) respectively. 

A general result from  \cite[\S 5]{radu2016} states that  free objects always exist in $\qacat(\qet)$, for any $\qet$, and they are isomorphic with term quantitative algebras for $\qet$. Moreover, such free objects are concretely identified for two relevant theories:

%One of the main results from \cite{radu2016} is the characterization of the free objects in $\qacat(\qet)$ as term quantitative algebras for $\qet$. 
%Moreover, such free objects are concretely identified for two relevant theories:
\begin{theorem}[{\cite[Cor 9.4 and 10.6]{radu2016}}]\label{theorems:radu:free:algebras}$ $
\vspace{-0.25cm}
\begin{itemize}
\itemsep-0.2cm
\item The free quantitative semilattice in $\qacat({\qet_{SL}})$ generated by a metric space $(X,d)$ is isomorphic to the metric space $\lfpset(X,d) = \big(\fpset(X), \haus(d) \big)$.
%, the metric space of finite non--empty subsets of $X$ endowed with the Hausdorff metric $\haus(d)$ where the $\oplus: \fpset X\times \fpset X\rightarrow \fpset X$ non--expansive operation is defined as binary union of sets.
\item The free quantitative convex algebra  in $\qacat({\qet_{CA}})$ generated by a metric space $(X,d)$ is isomorphic to the metric space $\ldset(X,d) = \big(\dset(X), \kant(d) \big)$.
%, the metric space of finitely supported probability distributions on $X$ endowed with the Kantorovich metric $\kant(d)$ where the $+_p:X\times X\rightarrow X$ non--expansive operation, for each $p\in[0,1]$, is defined as convex combination: $d_1 +_p d_2 = p\cdot d_1 + (1-p)\cdot d_2$.  
\end{itemize}
\end{theorem}

We remark that the above theorem from \cite{radu2016} falls short from a full presentation result stating the isomorphisms of categories  $\qacat(\qetcs) \cong \EM(\lfpset)$ and $\qacat(\qetca) \cong \EM(\ldset)$. This latter more general statement does indeed hold and can be obtained, with some minor extra work, from the technical machinery developed in  \cite{radu2016}  (see Footnote \ref{footnote:introduction}).

\section{The Monad $\lcset$ on the Category of Metric Spaces}
\label{section:4:proof:monad}
%!TEX root = paper.tex

%\ntodo{changes to discuss in previous sections:\\
%- problema nelle definizioni in sezione 3: a monad is a triple, it should be presented as such (and also we never give monads a name, why doing it for hyperspace monad?).\\
%- $\kompact$ is quite long, find an abbreviation? It could be $\pset_{k}$ or $\pset_{h}$? Also, it is not nice when we write $\kompact \wms$(used at end of the section) to denote $\wms$ applied below $\kompact$\\
%- come chiamare la monade $\cset$? finitely generated convex powerset of distributions monad? omettere finitely supported? manca il non-empty, anche nella definizione in sezione 3\\
%-usare conv invece di cc per convex closure?
%}

% In this section, we lift from $\Sets$ to $\Met$ the monad $\cset$ of finitely generated, non-empty convex sets of finitely supported distributions, presented in Definition \ref{def:set:cset}.

In this section we introduce a $\Met$ version of the $\Sets$ monad $\cset$, and we denote it with $\lcset$. The monad $\lcset$ is obtained by composing the Hausdorff lifting $\haus$ and the Kantorovich lifting $\kant$ introduced in the previous section. 
%In what follows we often  write $\haus\kant(d)$ for $\haus(\kant (d))$,
%\marginpar{this notation already much used in sections 2-3} the metric on the set $\hs(\dset(X))$ of non--empty compact sets of probability distributions on $X$, obtained by the Hausdorff lifting of the Kantorovich lifting of the metric $d$. 
%\marginpar{here it should be $\kompact$ used instead of $\hs$, since we are talking about sets and not metric spaces}
%First, we show that, given a metric space $(X,d)$, we can indeed define on $\cset(X)$ a metric based on such liftings. 
\begin{proposition}\label{cc:of:compact:is:compact}
%Let $(X,d)$ be a metric space and let $S\in \kompactd {\dset(X)} {\kant (d)}$ be a compact (with respect to $\kant(d)$) set of probability distributions. Then $\conv(S)$, the convex closure of $S$, is also compact, i.e., $\conv(S)\in \kompactd {\dset(X)} {\kant (d)}$. 
Let $(X,d)$ be a metric space and let $S\in \kompactd {\dset(X)} {\kant (d)}$. Then  $\conv(S)\in \kompactd {\dset(X)} {\kant (d)}$, i.e., the convex closure of $S$ is also compact.
%Furthermore the function $\conv: \kompactd {\dset(X)} {\kant (d)}\rightarrow \kompactd {\dset(X)} {\kant (d)}$ is non--expansive.
\end{proposition}
\begin{corollary}\label{corollary:useful}
Let $(X,d)$ be a metric space. If $S\in \cset(X)$ then $S \in \kompactd {\dset(X)} {\kant (d)}$. 
%, i.e., $S\in \kompactd {\dset(X)} {\kant (d)}$. %In other words $\cset(X)\subseteq  \hs(\dset(X))$.
\end{corollary}
%\begin{proof}
%Let $S = \conv\{ \Delta_1,\dots, \Delta_n\}$.  The set $\{ \Delta_1,\dots, \Delta_n\}\subseteq \dset(X)$ is finite, hence compact, thus $\{ \Delta_1,\dots, \Delta_n\}\in  \hs(\dset(X))$. Therefore $S$ is compact by Proposition \ref{cc:of:compact:is:compact}.
%\end{proof}
Corollary \ref{corollary:useful} implies that, given a metric space $(X,d)$, the collection $\cset(X)$ of finitely generated non--empty convex sets of probability distributions on $X$ can be endowed with the subspace metric of $\hs(\ldset(X,d))$, and therefore  $(\cset(X), \hk(d))$ is a metric space, with $\hk(d)=\haus(\kant(d))$. This observation leads to the following definition.

%We do so by composing the Hausdorff lifting $\haus$ with the Kantorovich lifting $\kant$. %In this section, we define such a lifting $\lcset$ of $\cset$, and we prove that $\lcset$ is indeed a monad.
%Given a metric space $(X,d)$, we define on $\cset (X)$ the metric $\haus(\kant (d))$, which first lifts $d$  to distributions, by applying $\kant$, and then to finitely generated convex sets, by applying $\haus$ on top of $\kant$. 
%This is well-defined since sets in $\cset (X)$ are compact, as required by the definition of $\haus$.  To see this, note that a finite set of elements of $\dset(X)$ is compact, being finite, and that the convex hull of a compact set is itself compact \cite{??}.
%\todo{citations and labels for referring to sections missing}
%In what follows, we generally write $\haus\kant (d)$ for $\haus(\kant (d))$.
%\todo{explain this notation somewhere?}
%The lifting $\lcset$ of the monad $\lcset$ on $\Sets$ to the category $\Met$.

\begin{definition}[Monad $\lcset$]
The \emph{finitely generated non--empty convex powerset of finitely supported distributions} monad $(\lcset, \eta^{\lcset}, \mu^{\lcset})$ on $\Met$ is defined as follows. Given an object $(X,d)$ in $\Met$, $\lcset(X,d) = \big(\cset(X), \hk(d) \big)$. % the elements of $\cset(X)$ (see Definition \ref{def:set:cset}) equipped with the distance $\hk(d)$. 
The action of $\lcset$ on morphisms, the monad unit $\eta^{\lcset}$, and the monad multiplication $\mu^{\lcset}$ are defined as for the $\Sets$ monad $\cset$ (Definition \ref{def:set:cset}).
\end{definition}

%\begin{definition}
%The monad $(\lcset, \eta, \mu)$ on $\Met$ is defined as follows:
%\begin{itemize}
%\item the functor $\lcset: \Met \to \Met$ is defined on objects as $\lcset (X,d) = (\cset (X), \hk (d))$, and on morphisms as $\lcset (f)= \cset f$;
%\item the unit $\eta$ and the multiplication $\mu$ are defined as those for $\cset$.
%\end{itemize}
%\end{definition}

The rest of this section is devoted to the proof that the above definition is well--specified, i.e., that $\lcset$ is indeed a monad on $\Met$.  First, one needs to verify that $\lcset$ is a functor on $\Met$. This follows immediately from the definition, Corollary \ref{corollary:useful}, and $\cset$ being a functor on $\Sets$. It then remains to verify that the unit $\eta^{\lcset}$ and the multiplication $\mu^{\lcset}$ of $\lcset$ are indeed morphisms in $\Met$ (i.e., they are non-expansive functions) and that they satisfy the monad laws of Definition \ref{monad:main_definition}. The fact that the laws are satisfied follows directly from the definitions $\mu^{\lcset}=\mu^\cset$ and $\eta^{\lcset}=\eta^\cset$ and the fact that $\cset$ is a monad on $\Sets$ (hence $\mu^\cset$ and $\eta^\cset$ satisfy the monad laws).  Then it only remains to verify that  $\eta^{\lcset}$ and  $\mu^{\lcset}$ are non--expansive. It is straightforward to verify that  $\eta^{\lcset}$ is an isometric (hence non--expansive) embedding of $(X,d)$ into  $\big(\cset(X), \hk(d) \big)$.
%\marginparblue{isometric embedding not defined}
%\begin{proposition}
%Let $(X,d)$ be a metric space in $\Met$. Then  $\eta^{\cset}_{(X,d)}: X\rightarrow \cset(X)$ is an isometric embedding (hence non--expansive).
%\end{proposition}
%\begin{proof}
%It follows directly form the definitions that $\eta^{\cset}$ is an isometric embedding, hence non--expansive.
%\end{proof}
%From the unit for the monads $\lpset$ and $\ldset$ being isometries, we derive that the unit of $\lcset$ is an isometry as well
%\[\hk (d) \funpair {\eta^{\lcset}} {\eta^{\lcset}} = d \]
%which then implies non-expansiveness of $\eta^{\lcset}$.
Proving that  $\mu^{\lcset}$ is non--expansive, instead, does not seem straightforward and requires some detailed calculations. We state this result as a theorem.

\begin{theorem}\label{mult:non--expansive:theorem}\label{thm:nemucset}
Let $(X,d)$ be a metric space in $\Met$. Then  $\eta^{\lcset}_{(X,d)}:  \lcset \lcset (X,d)\rightarrow \lcset (X,d) $
%$$
%S \mapsto \  \bigcup \big\{ \wms(\distr) \mid  \distr \in S  \big\} 
%$$
 is a non--expansive function, i,e., using functional notation,
%\[\hk (d) \funpair {\mu^{\lcset}} {\mu^{\lcset}} \funleq \hk\hk (d).\]
$\hk (d) \funpair {\mu^{\lcset}} {\mu^{\lcset}} \funleq \hk\hk (d)$.
\end{theorem}

\subsection{Sketch of the Proof of Theorem \ref{mult:non--expansive:theorem}}
The key result to prove is Lemma \ref{lem:newms}, stating that the weighted Minkowski sum function $\wms$  is non--expansive. This is obtained by exploiting a key property of the $\hk$ metric (see Lemma \ref{lem:hkconv}) called convexity.  It might well be that both these results have already appeared in the literature in some form or another or are known as folklore by specialists. We present here a direct proof. 

\begin{definition}[Convex metric]
Let $(X, \{+_p\}_{p\in(0,1)})$ be a convex algebra, i.e., a set $X$ equipped with operations $+_p:X^2\rightarrow X$ satisfying the axioms of Definition \ref{def:convex:set}. Let $d:X^2\rightarrow [0,1]$ be a metric on $X$. We say that $d$ is convex if $d(x_1 +_p x_2, y_1 +_p y_2) \leq  d(x_1,y_1) +_p d(x_2,y_2)$ holds for all $x_1,x_2,y_1,y_2\in X$, where $d(x_1,y_1) +_p d(x_2,y_2) = p\cdot d(x_1,y_1) + (1-p)\cdot d(x_2,y_2)$.
\end{definition}

It is well known that the Kantorovich metric $\kant(d)$ is convex.
%It is well known that the Kantorovich metric $\kant(d)$ on the collection of probability distributions $\dset{X}$ (which carries the structure of a convex algebra, see Example \ref{example:presentation:set:2}) is convex
%begin{proposition} \label{lem:kantconv}
%Let $(X,d)$ be a metric space. The metric $\kant (d)$ on the convex algebra $(\dset(X),  \{+_p\}_{p\in [0,1]})$, with $\Delta_1 +_p \Delta_2 = p_1\cdot \Delta_1 + (1-p_1)\cdot\Delta_2$, is convex.
%is a convex function, that is, for all $\Delta, \Delta', \Theta,\Theta' \in \dset (X)$ and for all $p\in[0,1]$ it holds
%\[\kant (d) (p\cdot \Delta+ (1-p)\cdot \Delta', p\cdot \Theta+ (1-p)\cdot \Theta')\leq p\cdot \kant (d) (\Delta,\Theta)+ (1-p)\cdot \kant (d)(\Delta',\Theta').\]
%\end{proposition}
The following lemma states that also the Hausdorff--Kantorovich metric $\haus\kant(d)$, on the collection $\cset(X)$ of non--empty finitely generated convex sets of distributions, which carries the structure of a convex semilattice (see Example \ref{example:presentation:set:3}) and thus also of a convex algebra, is convex.

%\begin{proposition} \label{lem:kantconv}
%Let $(X,d)$ be a metric space. The metric $\kant (d)$ on the convex algebra $(\dset(X),  \{+_p\}_{p\in [0,1]})$, with
 %$\Delta_1 +_p \Delta_2 = p_1\cdot \Delta_1 + (1-p_1)\cdot\Delta_2$, is convex. \todo{This proposition could be removed, the english above is sufficient, but we need to check if there are references to this prop in the main body.}

%is a convex function, that is, for all $\Delta, \Delta', \Theta,\Theta' \in \dset (X)$ and for all $p\in[0,1]$ it holds
%\[\kant (d) (p\cdot \Delta+ (1-p)\cdot \Delta', p\cdot \Theta+ (1-p)\cdot \Theta')\leq p\cdot \kant (d) (\Delta,\Theta)+ (1-p)\cdot \kant (d)(\Delta',\Theta').\]
%\todo{here I am implicitly using the $\mu$ for $\dset$, whenever I use $\cdot$}
%\end{proposition}

%The following lemma states that also the Hausdorff--Kantorovich metric $\haus\kant(d)$, on the collection $\cset(X)$ of non--empty finitely generated convex sets of probability distributions, which carries the structure of a convex semilattice and therefore also of a convex algebra (see Example \ref{example:presentation:set:3}), is convex.

\begin{lemma}\label{lem:hkconv}
Let $(X,d)$ be a metric space. The metric $\haus\kant (d)$ on the convex algebra $(\cset(X),  \{+_p\}_{p\in (0,1)})$, with
 $S_1 +_p S_2 = \wms( p_1 S_1 + (1-p_1)S_2)$, is convex.
\end{lemma}

Using the convexity of $\hk$ it is possible to prove that the $\wms$ function is non--expansive.

\begin{lemma}\label{lem:newms}
Let $(X,d)$ be a metric space. The function $\wms: \ldset(\lcset (X,d))\rightarrow \lcset (X,d)$ (see Definition \ref{def:set:cset}) is non--expansive, i.e.
%\[\hk (d)\big(  \wms(p_1S_1 + \dots +  p_nS_n),  \wms(q_1T_1 + \dots + q_nT_n) \big) \leq \kant \hk (d) \big(p_1S_1 + \dots p_nS_n, q_1T_1 + \dots q_nT_n  \big)\]
%or, equivalently and written more concisely using functional notation,
%\[\hk (d) \funpair \wms \wms \sqsubseteq \kant \hk (d). \]
$\hk (d) \funpair \wms \wms \sqsubseteq \kant \hk (d)$. 
\end{lemma}

%The proof of Theorem \ref{thm:nemucset} also relies on two simple properties of $\haus$ (Lemmas  \ref{lem:hausmon} and \ref{lem:hausfunpair}).

Lastly, we state the following two useful properties of the Hausdorff lifting. %Lemma  \ref{lem:hausmon} states that $\haus$ is monotonic and Lemma \ref{lem:hausfunpair} states a property regarding the interaction between $\haus$ and non--expansive functions between metric spaces.

%concluding with the proof that $\mu^{\cset}$ is non-expansive, we need two simple lemmas on the Hausdorff lifting.

%Lemma  \ref{lem:hausmon} shows that $\haus$ is monotonic. And  Lemma \ref{lem:hausfunpair} states a property regarding the interaction between $\haus$ and non--expansive functions between metric spaces.

\begin{proposition}\label{lem:hausmon}
Let $d,d'$ be two metrics over $X$ such that $d \funleq d'$.
Then $\haus (d) \funleq \haus (d')$.
\end{proposition}

\begin{proposition}\label{lem:hausfunpair}
Let $(X,d_{X})$ and $(Y,d_{Y})$ be metric spaces, let $f: X\to Y$ with $d_{X}=d_{Y}\funpair f f$ (i.e., $d_X(x_1,x_2) = d_Y(f(x_1),f(x_2)$). %, and let $g: \upset X \to \upset Y$. 
Then $\haus (d_{X})  =\haus (d_{Y}) \funpair {\hs (f)}{\hs (f)}$.
%Then $\haus (d') \funpair g g =\haus (d) \funpair {g\circ \upset f}{g\circ  \upset f}$.
\end{proposition}

%%%%%%%%%%%%%%%%

%We can now finally present the proof of Theorem \ref{thm:nemucset}.
\begin{proof}[Proof of Theorem \ref{thm:nemucset}]
We need to show that $\hk (d) \funpair {\mu^{\lcset}} {\mu^{\lcset}} \funleq \hk\hk (d)$.

Since $\hs$ is a monad on $\Met$ (Definition \ref{def:mett:compact}), $\mu^{\hs}$ is non-expansive, i.e.,
$\haus (d)
\funpair {\mu^{\hs}}{\mu^{\hs}}
\funleq 
\haus \haus (d)$. By applying this to the metric $\kant (d)$, we derive
\begin{equation}\label{eq:nemucset1}
\hk (d)
\funpair {\mu^{\hs}}{\mu^{\hs}}
\funleq 
\haus \hk (d).
\end{equation}
By definition $\mu^{\lcset}=\mu^{\hs}\circ \hs(\wms)$ (i.e., 
$
S\ \mapsto \  \bigcup \big\{ \wms(\distr) \mid  \distr \in S  \big\} 
$) and therefore: 
\begin{align*}
\hk (d) \funpair {\mu^{\lcset}}{\mu^{\lcset}}
&=\hk (d)
\funpair {\mu^{\hs}\circ \hs(\wms)}{\mu^{\hs}\circ \hs(\wms)}\\
&=\hk (d)
\funpair {\mu^{\hs}}{\mu^{\hs}} \funpair {\hs (\wms)}{\hs( \wms)} 
\end{align*}
Thus, by (\ref{eq:nemucset1})
we can derive
\begin{equation}\label{eq:nemucset2}
\hk (d) \funpair {\mu^{\lcset}}{\mu^{\lcset}}
\funleq \haus \hk (d)\funpair {\hs( \wms)}{\hs( \wms)}.
\end{equation}
Moreover, by the non-expansiveness of $\wms$ (Lemma \ref{lem:newms}), we know that
\[\hk (d) \funpair \wms \wms \funleq \kant \hk (d)\]
which implies by the monotonicity of $\haus$ (Proposition \ref{lem:hausmon}) that
\begin{equation}\label{eq:nemucset31}
\haus(\hk (d) \funpair \wms \wms) \funleq \hk \hk (d).
\end{equation}
By Proposition \ref{lem:hausfunpair}, we can rewrite the left-hand term of (\ref{eq:nemucset31}) as follows
\[ \haus (\hk (d)\funpair { \wms}{\wms})= \haus\hk (d)\funpair {\hs( \wms)}{\hs( \wms)}\]
and thus we derive from (\ref{eq:nemucset31}):
\begin{equation}\label{eq:nemucset3}
\haus\hk (d)\funpair {\hs(\wms)}{\hs( \wms)}\funleq  \hk \hk (d).
\end{equation}
Lastly, by (\ref{eq:nemucset2}) and (\ref{eq:nemucset3}):
$\hk (d) \funpair {\mu^{\lcset}}{\mu^{\lcset}}
\funleq \haus \hk (d)\funpair {\hs (\wms)}{\hs( \wms)}
\funleq \hk \hk (d). $
\end{proof}

\section{Presentation of the Monad $\lcset$}
\label{section:5:proof:presentation}
%!TEX root = paper.tex

In this section we present the main result of this work and show that the monad $\lcset$ on $\Met$, introduced in Section \ref{section:4:proof:monad}, is presented by quantitative convex semilattices.
\begin{definition}\label{def:qcs}
The quantitative equational theory of \emph{quantitative convex semilattices}, denoted by  $\qetcs$, is the quantitative theory over the signature $\sigcs= (\{ \oplus \}\cup \{ +_p \}_{p\in (0,1)})$ of convex semilattices induced by the following set quantitative inferences:
%\marginpar{Attenzione a $(0,1)$ vs. $[0,1]$}
\vspace{-0.25cm}
\begin{itemize}
\itemsep-0.2cm
\item the quantitative inferences ($A$), ($C$), ($I$) and ($H$) inducing the quantitative theory of semilattices (see Definition \ref{def:semilattices:met}), 
\item the quantitative inferences ($A_p$), ($C_p$), ($I_p$), and ($K$) inducing the quantitative theory of convex algebras (see Definition \ref{def:convex:met}), 
\item for every $p\in (0,1)$, the quantitative inference ($D$) $\emptyset \vdash x+_p (y \oplus z) =_{0} (x+_p y)\oplus (x+_p z)$.
\end{itemize}
\end{definition}
%\marginpar{VVV:here I put ``inducing'' instead of ``generating''}
%Equivalently, the quantitative equational theory of convex semilattices is the theory generated by the (quantitative analogues of the) axioms of convex semilattices, rules (H) and (K).

The following is the main result of this work.

\begin{theorem}\label{thm:main}
The quantitative equational theory $\qetcs$ of quantitative convex semilattices is a presentation of the monad $\lcset$, that is, $\qacat(\qetcs)\cong \EM(\lcset)$.
\end{theorem}

As one direct corollary of this general statement we automatically get the following result (cf. with Theorem \ref{theorems:radu:free:algebras}) characterising free quantitative convex semilattices, which, by \cite[\S 5]{radu2016}, are in turn isomorphic to term quantitative algebras for $\qetcs$.
%\marginparblue{scrivere qui che catturiamo la term algebra, rifacendosi ai risultati di Radu (magari con ``is isomorphic to: lista'')}
\begin{corollary}
The free quantitative algebra in $\qacat({\qetcs})$ generated by a metric space $(X,d)$ is isomorphic to
%\vspace{-0.25cm}
%\begin{itemize}
%\itemsep-0.2cm
%\item 
$\lcset{(X,d)}$, the metric space of finitely generated convex sets of probability distributions metrized by the Hausdorff--Kantorovich metric $\haus\kant (d)$.
%\item the term quantitative algebra  for quantitative convex semilattices.
%\end{itemize}
\end{corollary}

%The rest of this section is devoted to the proof of Theorem \ref{def:qcs}.

%\subsection{Sketch of the Proof of Theorem \ref{thm:main}.}

We prove Theorem \ref{thm:main} by explicitly defining a pair of functors
$\calf:  \EM(\lcset) \to \qacat(\qetcs)$ and $\calg: \qacat(\qetcs)\to \EM(\lcset)$
and proving that  they are isomorphisms of categories, i.e., that $\calg \circ \calf = id_{ \EM(\lcset)}$ and $\calf \circ \calg = id_{\qacat(\qetcs)}$.
In the following sections, we exhibit such functors and show that they are well-defined isomorphisms.
%n Subsection \ref{sec:F} we define the functor $\calf$ and prove that it is well-defined. 
%Analogously, in Section \ref{sec:G} we define $\calg$ and prove that it is well-defined. 
%We conclude with the proof that they are indeed isomorphisms of categories in Subsection \ref{sec:iso}.
%\marginparblue{check id vs. 1!! Meglio id}

\subsection{The functor $\calf:\EM(\lcset) \to \qacat(\qetcs)$}\label{sec:F}

%\textbf{The functor $\calf:\EM(\lcset) \to \qacat(\qetcs)$.}

Recall from Definition \ref{def:algebra-of-a-monad} that an object in  $\EM(\lcset)$ is a structure $((X,d),\alpha)$ where $(X,d)$ is a metric space and $\alpha: (\cset (X),\hk(d)) \to (X,d)$ is a non-expansive function satisfying $ \alpha \circ  \eta^{\lcset}_X = id_X$ and $ \alpha \circ \lcset \alpha= \alpha \circ \mu^{\lcset}_X$. A morphism $f: ((X,d_X),\alpha_X) \rightarrow ((Y,d_Y),\alpha_Y)$ in  $\EM(\lcset)$ is a non--expansive function $f: X \rightarrow Y$ such that  $f\circ \alpha_X = \alpha_Y \circ \lcset(f)$.
%\marginparblue{check equations for algebra}

\begin{definition}[Functor $\calf$] We define $\calf:  \EM(\lcset) \to \qacat(\qetcs)$ as follows:
\vspace{-0.2cm}
\begin{itemize}
\itemsep-0.2cm
\item on objects: $\calf((X,d),\alpha)= (X, \sigcs^{\alpha}, d)$\\
%\{\cplus^{\alpha}\} \cup \bigcup_{p \in [0,1]}\{ \pplus p^\alpha\}
with $\sigcs^{\alpha} = (\{\cplus^{\alpha}\} \cup\{ \pplus p^\alpha\}_{p\in(0,1)})$ the interpretation of the convex semilattice operations $\cplus$ and $\pplus p$ as $x_{1} \cplus^{\alpha} x_{2} = \alpha(\conv\{\dirac {x_{1}}, \dirac {x_{2}}\})$ and $x_{1} \pplus p^\alpha x_{2}= \alpha(\{p x_{1} + (1-p) x_{2}\})$, 
\item on morphisms: $\calf(f)=f$, with $f:X\rightarrow Y$ seen as a non--expansive map from $X$ to $Y$.
\end{itemize}
\end{definition}

We now prove that the functor $\calf$ is well-defined.  First, on objects, we need to show that $\calf((X,d),\alpha)$ is indeed a quantitative algebra satisfying the quantitative inferences of the theory $\qetcs$. To show that $(X, \sigcs^{\alpha}, d)$ is a quantitative algebra (Definition \ref{quantitative:algebra:def}), since $(X,d)$ is a metric space, we only need to verify that the operations $\cplus^{\alpha}$ and $\pplus p^\alpha$ are non--expansive. %We state this as a lemma.
\begin{lemma}\label{non--expansive:lemma5.1.1}
 The operations $\cplus^{\alpha}$ and $\pplus p^\alpha$, for all $p\in(0,1)$, are non--expansive. 
\end{lemma}
\begin{proof}
Using functional notation we have 
$\cplus^{\alpha} = \alpha \,\circ\, \conv \,\circ \,\fpset \eta^{\dset}_{X}\, \circ\, ( \lambda x_1,x_2 . \{x_1,x_2\})$.
The function $\alpha$ is non--expansive by assumption. $\fpset \eta^{\dset}_{X}$ is non-expansive by $\lfpset$ and $\ldset$ being monads on $\Met$. The functions $ \lambda x_1,x_2 . \{x_1,x_2\}:{(X,d)}\times (X,d)\rightarrow \lpset{(X,d)}$ and $\conv: \lpset\ldset(X,d) \rightarrow \lcset(X,d)$ are non--expansive as well.
Hence $\cplus^{\alpha}$ is non--expansive as composition of non--expansive maps. Similarly, we have $\pplus p^\alpha= \alpha   \circ  \eta^{\pset}_{\dset(X)} \circ \big( \lambda {x_1,x_2}. (p x_1 + (1-p)x_2 )\big)$ and all operations involved are non--expansive.
%
%Using functional notation we have 
%$\cplus^{\alpha} = \alpha \circ \conv \circ ( \lambda_{<d_1,d_2>}. \{d_1,d_2\}) \circ < \delta,\delta >$.
%The function $\alpha$ is non--expansive by assumption. $\delta: (X,d)\rightarrow \ldset{(X,d)}$ is the unit of the monad $\ldset$ and is therefore non--expansive. The functions $ \lambda_{<d_1,d_2>}. \{d_1,d_2\}:(\ldset{(X,d)})^2\rightarrow \lpset\lcset{(X,d)}$ and $\conv: \lpset\ldset(X,d) \rightarrow \lcset(X,d)$ are non--expansive as well.
%Hence $\cplus^{\alpha}$ is non--expansive as composition of non--expansive maps. Similarly, we have $\pplus p^\alpha= \alpha   \circ \big( \lambda_{<d_1,d_2>}. \{p\cdot d_1 + (1-p)d_2 \}\big) \circ  < \delta,\delta > $ and all operations involved are non--expansive.
\end{proof}

As $\calf((X,d),\alpha)$ is a quantitative algebra, it satisfies all the quantitative inferences of Definition \ref{def:qet}. It only remains to show that the quantitative inferences of the theory $\qetcs$ (Definition \ref{def:qcs}) are also satisfied. For each of the quantitative inferences  ($A$, $C$, $I$, $A_p$, $C_p$, $I_p$, $D$), which are of the form  $\emptyset \vdash s=_0 t$, we  need to show that the equality $s=t$ holds (universally quantified) in $(X, \sigcs^{\alpha}, d)$. This amounts to showing that the algebra $(X, \sigcs^{\alpha})$ (with the metric $d$ forgotten) is a model of the equational theory of convex semilattices (Definition \ref{def:convexsemilattices:set}). This proof has no specific metric--theoretic content and is omitted here. Thus, it only remains to show that the quantitative inferences ($H$) and ($K$) are satisfied. 

\begin{lemma}[H]
$\big\{x_1=_{\epsilon_1} y_{1},     x_2=_{\epsilon_2} y_{2}\big\} \models_{\calf((X,d),\alpha)}  x_1 \oplus x_2 =_{\max(\epsilon_1,\epsilon_2)} y_{1}\oplus y_{2}$.
\end{lemma}
\begin{proof}
The quantitative inference ($H$) is equivalent (i.e., mutually derivable in presence of the others deductive rules of Definition   \ref{def:qet}) with the (NExp) deductive rule. This means that ($H$) holds in $\calf((X,d),\alpha)$ because the operation $\cplus^{\alpha}$ is non--expansive (Lemma \ref{non--expansive:lemma5.1.1}). 
\end{proof}

\begin{lemma}[K]
$x_{1} =_{\epsilon_{1}} y_{1}, x_{2} =_{\epsilon_{2}} y_{2}  \models_{\calf((X,d),\alpha)}  x_{1} \pplus p x_{2} =_{p \cdot \epsilon_{1} + (1-p) \cdot \epsilon_{2}} y_{1} \pplus p y_{2}$.
\end{lemma}
\begin{proof}
For arbitrary $x_1,x_2,y_1,y_2\in X$, assume $d ({x_{1}},{y_{1}}) \leq \epsilon_{1}$ and  $d({x_{2}},{y_{2}}) \leq \epsilon_{2} $. 
Then
\begin{align*}
d(x_{1} \pplus p^{\alpha} x_{2}, y_{1} \pplus p^{\alpha} y_{2})
&= d(\alpha(\{p {x_{1}} +(1-p) {x_{2}}\}), \alpha(\{p {y_{1}}+(1-p) {y_{2}}\})\\
&\leq \hk (d)(\{p {x_{1}} +(1-p) {x_{2}}\}, \{p {y_{1}}+(1-p) {y_{2}}\})\tag{$\alpha$ non-exp.}\\
&= \kant (d)(p {x_{1}} +(1-p) {x_{2}}, p {y_{1}}+(1-p) {y_{2}})\\
&\leq p \cdot d(x_{1},y_{1}) + (1-p) \cdot d(x_{2}, y_{2})\tag{the metric $\kant (d)$ is convex}\\
&\leq p \cdot \epsilon_{1} + (1-p) \cdot\epsilon_{2} \tag*{\hspace*{\fill} \qedhere}
\end{align*}
\end{proof}

Hence $\calf$ is well--defined on objects. It remains to verify that $\calf$ is well defined on morphisms. Let $f:((X,d), \alpha)\to ((Y,d'), \beta)$ be a morphism in  $\EM(\lcset)$. We need to verify that 
$\calf(f)$ is a morphisms in  $\qacat(\qetcs)$, i.e.,  a non--expansive homomorphism of convex semilattices (see Definition \ref{quantitative:algebra:def}).  Since by definition $\calf(f)=f$, the function $\calf(f)$ is non--expansive. It remains to verify that it is a homomorphism. This proof has no specific metric--theoretic content and we omit it here.

\subsection{The functor $\calg: \qacat(\qetcs) \to \EM(\lcset)$}\label{sec:G}

Recall that an object in  $\qacat(\qetcs)$ is a quantitative convex semilattice  $\alga =(X, \sigcs^{\alga},d)$, with $\sigcs^{\alga}=(\{\cplus^{\alga}\} \cup \{ \pplus p^\alga\}_{p\in(0,1)})$. Also, recall from Example \ref{example:presentation:set:3} that there is an isomorphism $\ics$ mapping elements of $\cset(X)$ to equivalence classes of convex semilattice terms in ${\terms X {\sigcs}}_{/\etcs}$. 
Let us define $\nf: \cset(X) \to \terms X {\sigcs}$ as a choice function, mapping each $S \in \cset(X)$ to one representative of the equivalence class $\ics(S)$. This allows us to uniquely write down each $S\in\cset(X)$ as a convex semilattice term: 
\[\nf (S) = \bigcplus_{\Delta\in \ub(S) } (\bigpplus_{x \in \support (\distr)} \distr(x) \, x).\]
%Given a distribution $\Delta$, we sometimes write $\nf (\Delta)$ instead of $\nf (\{\Delta\})$. 
With 
%\marginparred{check that $\nf (\{\Delta\})$ is used and not $\nf (\Delta)$ VVV: it's okay, I checked. However this notation is used in one proof in appendix, but I write it there} 
abuse of notation, we have used the letter $X$ to range both over a set of variables and the carrier of $\alga$. 
By interpreting each variable $x$ with the corresponding element $x\!\in\! X$ of $\alga$, and by homomorphic extension, we get that each term $t\!\in\!  \terms X {\sigcs}$ can be interpreted as an element $t^\alga$ of $\alga$, and in particular $(\nf(S))^\alga$ denotes an element of $\alga$  for each $S\in \cset(X)$.

\begin{definition}[Functor $\calg$] We specify $\calg:  \qacat(\qetcs) \rightarrow \EM(\lcset) $ as follows:
\vspace{-0.2cm}
\begin{itemize}
\itemsep-0.2cm
\item on objects $\alga =(X, \sigcs^{\alga},d)$, we define $\calg(\alga) = ((X,d), \alpha)$,\\
with $\alpha: (\cset (X),\hk(d)) \to (X,d)$ defined as: $\alpha (S)= (\nf(S))^{\alga}$,
\item on morphisms (i.e., non-expansive homomorphisms) we define $\calg(f)=f$.
\end{itemize}
\end{definition}

In order to prove that $\calg$ is well-defined on objects, we have to show that indeed $((X,d), \alpha)$ is an Eilenberg-Moore algebra for $\lcset$, which amounts to proving the following lemma.
\begin{lemma}\label{lem:Gwd}
Let $\calg(\alga) = ((X,d), \alpha)$, for $\alga=(X, \sigcs^{\alga},d)\in \qacat(\qetcs)$.
\vspace{-0.2cm}
\begin{enumerate}
\itemsep-0.2cm
\item $(X,\alpha)$ is an Eilenberg-Moore algebra for $\cset$ in $\Sets$, i.e., $\alpha \circ \eta^{\cset}=id$ and $\alpha \circ \cset\alpha= \alpha \circ \mu^{\cset}$. 
\item $\alpha$ is a morphism in $\Met$, i.e., $\alpha$ is a non-expansive map: $d\funpair \alpha \alpha \funleq \hk (d)$.
\end{enumerate}
\end{lemma}
\begin{proof}
The proof of the first point does not have any specific metric--theoretic content and is omitted here. For the second point, let $S,T\in \cset(X)$. By the definition of $\alpha$, we have $d(\alpha(S), \alpha(T)) 
=  d((\nf(S))^{\alga}, (\nf(T))^{\alga})$.
%&=d\Big(\bigcplus_{(\sum_{i} p_{i} x_{i})\in \ub(S) }^{\alga} \big(\bigpplus_{i}^{\alga} p_{i} \alpha(x_{i})\big),  
%&\bigcplus_{(\sum_{j} q_{j} y_{j})\in \ub(T) }^{\alga} \big(\bigpplus_{j}^{\alga} q_{j} \alpha(y_{j})\big)\Big)
As stated in Lemma \ref{lem:Ghk} below, it is possible to derive in $\qetcs$ the quantitative inference
%\begin{equation}
%\label{eq:Ghk}
%\bigcup_{(x,y) \in \bigcup_{\Delta\in \ub(S),\Theta\in\ub(T)} \support(\Delta)\times \support(\Theta)}\{x=_{d(x,y)} y\}
%\vdash 
%\nf(S) 
%%\bigcplus_{(\sum_{i} p_{i} x_{i})\in \ub(S) }^{\alga} \big(\bigpplus_{i}^{\alga} p_{i} \alpha(x_{i})\big)
%=_{\hk(d)(S,T)}  
%\nf(T)
$$
\bigcup_{(\Delta,\Theta)\in\ub(S)\times \ub(T)} (\bigcup_{(x,y) \in \support(\Delta)\times \support(\Theta)}\{x=_{d(x,y)} y\})
\vdash 
\nf(S) 
%\bigcplus_{(\sum_{i} p_{i} x_{i})\in \ub(S) }^{\alga} \big(\bigpplus_{i}^{\alga} p_{i} \alpha(x_{i})\big)
=_{\hk(d)(S,T)}  
\nf(T)
$$
%\bigcplus_{(\sum_{j} q_{j} x_{j})\in \ub(T) }^{\alga} \big(\bigpplus_{j}^{\alga} q_{j} \alpha(x_{j})\big\]
%\end{equation}
%\todo{perhaps remove numbered equation since we are not referring to it with a number anymore.}
which, since $\alga$ is a model of $\qetcs$, is thereby satisfied by $\alga$. Since all the premises of the inference hold in $\alga$, we conclude that $d((\nf(S))^\alga, (\nf(T))^\alga) \leq  \hk(d)(S,T)$ and, therefore, $d\funpair \alpha \alpha \funleq \hk (d)$ holds, as desired.
\end{proof}

The following technical lemma  is critically used in the proof of Lemma \ref{lem:Gwd}(2) above. Note that its statement is purely syntactic as it deals with derivability in the deductive apparatus of quantitative equational theories (Definition \ref{def:qet}).

\begin{lemma}\label{lem:Ghk}
Let $(X,d)$ be a metric space and let $S,T \in \cset(X)$. Then we have in $\qetcs$:
\[
\bigcup_{(\Delta,\Theta)\in\ub(S)\times \ub(T)} (\bigcup_{(x,y) \in \support(\Delta)\times \support(\Theta)}\{x=_{d(x,y)} y\})
\vdash 
\nf(S) 
%\bigcplus_{(\sum_{i} p_{i} x_{i})\in \ub(S) }^{\alga} \big(\bigpplus_{i}^{\alga} p_{i} \alpha(x_{i})\big)
=_{\hk(d)(S,T)}  
\nf(T)
\]
\end{lemma}
\begin{proof}[Proof Sketch]
First, we derive the following useful quantitative inference dealing with the case of  $S=\{\Delta\}$ and $T=\{\Theta\}$ being singletons, so that $\hk(d)(S,T) = \kant(d)(\Delta,\Theta)$.  Let $(X,d)$ be a metric space and let $\Delta, \Theta \in \dset(X)$. Then the following is derivable in $\qetcs$:
\[\bigcup_{(x,y)\in \support (\Delta) \times \support (\Theta)}\{x=_{d(x,y)} y\}\vdash  \nf(\{\Delta\}) =_{\kant(d)(\Delta,\Theta)} \nf(\{\Theta\}).\]
To construct this derivation we take an optimal coupling $\omega$ of $\Delta$ and $\Theta$ (see Definition \ref{def:kantorovich:lifting}) witnessing the Kantorovich distance $\kant(d)(\Delta,\Theta)$ and then use the information provided by $\omega$ to construct a syntactic derivation where only the quantitative inferences ($A_p$, $C_p$, $I_p$ and $K$)  of the quantitative theory of convex algebras are used. 
The construction of this derivation follows analogously to the completeness result for quantitative convex algebras from \cite{radu2016}.

Secondly, we calculate the $\hk(d)(S,T)$ distance between $S$ and $T$.
$$
\hk(d)(S,T) = \max\big\{  \sup_{\Delta\in S}\inf_{\Theta \in T}\kant(d)(\Delta,\Theta)  \ \     , \ \  \sup_{\Theta\in T}\inf_{\Delta \in S}\kant(d)(\Delta,\Theta)    \big\}.
$$
By compactness arguments, the $\inf$ and $\sup$ are always attained. Hence this calculation involves distances $\kant(d)(\Delta_i,\Theta_j)$ between a finite number of elements $\Delta_i\in S$ and $\Theta_j \in T$, for $0\leq i \leq n$ and $0\leq j \leq m$. 
%By Proposition \ref{prop:removingconv} 
Since the equation $x \oplus y = x \oplus y \oplus ( x +_p y)$ holds in all convex semilattices, 
we can derive in the theory of convex semilattices the equalities: $\nf(S) = \nf(S) \oplus \nf(\{\Delta_1\}) \oplus \dots \oplus \nf(\{\Delta_n\})$ and $\nf(T) = \nf(T) \oplus \nf(\{\Theta_1\}) \oplus \dots \oplus \nf(\{\Theta_m\})$. For each of the pairs $(\Delta_i, \Theta_j)$ appearing in the expressions above we can derive, as described above, the quantitative equation  $\nf(\{\Delta_i\}) =_{\kant(d)(\Delta_i,\Theta_j)} \nf(\{\Theta_j\})$.  The calculation of $\hk(d)(S,T)$ can then be mimicked syntactically to derive the quantitative equation $\nf(S) 
%\bigcplus_{(\sum_{i} p_{i} x_{i})\in \ub(S) }^{\alga} \big(\bigpplus_{i}^{\alga} p_{i} \alpha(x_{i})\big)
=_{\hk(d)(S,T)}  
\nf(T)$ by only using the quantitative inferences ($A$, $C$, $I$ and $H$)  of quantitative semilattices. This follows analogously to the completeness result for quantitative semilattices from \cite{radu2016}.
\end{proof}

It remains to verify that the functor $\calg$ is well-defined on morphisms.  To see this, take $f:X\to Y$ a non-expansive homomorphism of quantitative algebras $\alga =(X, \sigcs^{\alga},d)$
and $\algb =(Y, \sigcs^{\algb},d')$ in $\qacat(\qetcs)$.
Then $f$ is an arrow in $\Met$, being non-expansive. We therefore only need to show that $f$ is also a morphism of Eilenberg-Moore algebras (see Definition \ref{def:algebra-of-a-monad}) i.e., that $ f\circ \alpha = \beta \circ \lcset(f)  $. The verification of this equality involves no specific metric--theoretic considerations, and is therefore omitted.
%\marginparblue{Eilenberg-Moore con trattino to check}
%\marginparblue{Write in the introduction the note about proofs being in appendix}

\subsection{The isomorphism}\label{sec:iso}

It remains to prove that the functors $\calf:  \EM(\lcset) \to \qacat(\qetcs)$ and $\calg:  \EM(\lcset) \to \qacat(\qetcs)$ define an isomorphism between the categories 
$\EM(\lcset)$ and  $\qacat(\qetcs)$. This means proving that $\calg \circ \calf = id_{\EM(\lcset)}$ and $\calf \circ \calg = id_{\qacat(\qetcs)}$. On morphisms,  by definition we have $\calg \circ \calf(f) = f = \calf \circ \calg(f)$. Hence the identities trivially hold true. The proofs regarding the identities on objects require only routine verifications, unfolding definitions, not involving any specific metric--theoretic content and therefore we omit them here.

%The isomorphism $\acat(\etcs)\cong \EM(\cset)$ is given by the pair of functors
%\begin{align*}
%&F: \EM(\cset) \to \acat(\etcs) \\
%&G: \acat(\etcs) \to \EM(\cset)
%\end{align*}
%defined as follows.
%For $(X,\alpha)\in \EM(\cset)$,
%$F(X,\alpha) = (X, \{\cplus^{\alpha}\} \cup \bigcup_{p \in [0,1]}\{ \pplus p^\alpha\})$
%with $(X, \{\cplus^{\alpha}\} \cup \bigcup_{p \in [0,1]}\{ \pplus p^\alpha\})$ the convex semilattice with operations interpreted as:
%\begin{align*}
%x_{1} \cplus^{\alpha} x_{2}&= \alpha(\conv \{\dirac x, \dirac y\})\\
%x_{1} \pplus p^\alpha x_{2}&= \alpha(\{p x_{1} + (1-p) x_{2}\}).
%\end{align*}

\section{Conclusions}
%!TEX root = paper.tex

We have introduced the $\Met$ monad $\lcset$ of non--empty convex sets of distributions equipped with the Hausdorff-Kantorovich distance, and we have proved that $\lcset$ is presented by the quantitative equational theory $\qetcs$ of quantitative convex semilattices.  This result provides the basis
%is the main building block 
for a foundational understanding of equational reasoning about program distances in processes combining nondeterminism and probabilities, as in  bisimulation and trace metrics \cite{DengCPP06, GeblerLT16, TangB18, BacciBLMTB19, Cast18}. This opens several directions for future research.

%DengCPP06 https://hal.inria.fr/inria-00201087/document 
%TangB18 https://drops.dagstuhl.de/opus/volltexte/2018/9547/pdf/LIPIcs-CONCUR-2018-9.pdf
%BacciBLMTB19 https://drops.dagstuhl.de/opus/volltexte/2019/10911/pdf/LIPIcs-CONCUR-2019-9.pdf
%GeblerLT16 https://lmcs.episciences.org/2627/pdf
%
%https://link.springer.com/chapter/10.1007/978-3-319-06880-0_10 ?
%https://users.soe.ucsc.edu/~luca/papers/08/j9-game-refinements.pdf ?
%Several ad-hoc equational theories

For instance, one interesting line of research is to examine the axiomatizations of bisimulation equivalences and metrics for nondeterministic and probabilistic programs (or process algebras) that have been proposed in the literature \cite{MOW03, BS01, DP07, Andova99, DArgenioGL14}. The quantitative equational framework of quantitative convex semilattices provides a novel tool for comparing and further developing the existing works.

It is also important to explore variants of the $\Met$ monad $\lcset$ such as, for instance, the one that also includes the empty set. 
These are needed to model program observations such as termination.
%MODIFIED
%It is also important to explore variants of the $\Met$ monad $\lcset$ such as, for example, the one that also includes the emptyset. 
%These are useful to model program observations such as, e.g., termination. 
Following the ideas presented in \cite{BSV19}, these variants can be explored via the lift monad $(\cdot +1)$ and its quotients described by equational theories over the signature of convex semilattices extended with a new constant symbol. A systematic study of these quotients is a promising direction for future work. 
%ADDED THIS
Applications to up-to techniques for bisimulation metrics \cite{CPV16, BonchiKP18} could then be pursued as well.

Lastly, it is natural to ask if the monad $\lcset$, and its presentation, can be obtained as a general categorical composition of the hyperspace monad $\hs$ and the distribution monad $\ldset$. The recently announced results of Goy and Petrisan  \cite{Petrisan20}, if applicable to the category $\Met$, might shed some light on this question.

\newpage

\newpage
\bibliography{biblio} 

\newpage
\appendix

\section{Proofs for Section \ref{section:4:proof:monad}}

We first recall some basic properties of suprema and infima.
\begin{proposition}\label{appendix:basic:infsup:properties}
Let $X,Y$ be sets and $f:X\rightarrow\mathbb{R}$ and $f^\prime:Y\rightarrow\mathbb{R}$ be arbitrary bounded functions. The following properties  hold:
\begin{align*}
%\item[(SUP)] $
&\text{(SUP)} &&\sup_{x\in X, y\in Y}\big(\bindistr p  {f(x)} {f'(y)}\big) =\bindistr p {(\sup_{x\in X} f(x))}  {(\sup_{y\in Y} f'(y))} \\
&\text{(INF)}&&\inf_{x\in X, y\in Y}\big(\bindistr p  {f(x)} {f'(y)}\big)=\bindistr p {(\inf_{x\in X} f(x))}  {(\inf_{y\in Y} f'(y))}
%& \text{(SUPINF$\leq$)}&&\text{if $X \leq Y$ then $\sup X \leq \sup Y$ and $\inf X \leq \inf Y$}
%& \text{(INF$\leq$)}&&\text{if $X \leq Y$ then $\inf X \leq \inf Y$}\\
\end{align*}
%\marginpar{removed property (SUPINF$\leq$), also to check the ``arbitrary bounded functions''. Also, maybe ``ordered sets'' not needed now}
%where, for sets $X,Y$ with ordered elements,  $X \leq Y$ holds if for all $x \in X$ there is a $y \in Y$ such that $x\leq y$, and 
%for all $y \in Y$ there is a $x \in X$ such that $x\leq y$.
\end{proposition}

It is well-known that the Kantorovich distance is convex. We present here an easy proof of this result.

\begin{proposition} \label{prop:kantconv}
Let $(X,d)$ be a metric space. The metric $\kant (d)$ on the convex algebra $(\dset(X),  \{+_p\}_{p\in (0,1)})$, with
 $\Delta_1 +_p \Delta_2 = p_1\cdot \Delta_1 + (1-p_1)\cdot\Delta_2$, is convex.
%  \todo{This proposition could be removed, the english above is sufficient, but we need to check if there are references to this prop in the main body.}
%
%is a convex function, that is, for all $\Delta, \Delta', \Theta,\Theta' \in \dset (X)$ and for all $p\in[0,1]$ it holds
%\[\kant (d) (p\cdot \Delta+ (1-p)\cdot \Delta', p\cdot \Theta+ (1-p)\cdot \Theta')\leq p\cdot \kant (d) (\Delta,\Theta)+ (1-p)\cdot \kant (d)(\Delta',\Theta').\]
%\todo{here I am implicitly using the $\mu$ for $\dset$, whenever I use $\cdot$}
\end{proposition}
\begin{proof}[Proof of Proposition \ref{prop:kantconv}.]
%We show that
%\[\kant d (p\cdot \Delta+ (1-p)\cdot \Delta', p\cdot \Theta+ (1-p)\cdot \Theta')\leq p\cdot \kant (\Delta,\Theta)+ (1-p)\cdot \kant d(\Delta',\Theta')\]

We need to show that \[\kant (d) (p\cdot \Delta+ (1-p)\cdot \Delta', p\cdot \Theta+ (1-p)\cdot \Theta')\leq p\cdot \kant (d) (\Delta,\Theta)+ (1-p)\cdot \kant (d)(\Delta',\Theta').\]

Take any pair of couplings $\omega'\in Coup(\Delta,\Theta), \omega''\in Coup(\Delta',\Theta')$. Then the distribution $\omega= p\cdot \omega' + (1-p)\cdot \omega''$ is a coupling for $(p\cdot \Delta+ (1-p)\cdot \Delta', p\cdot \Theta+ (1-p)\cdot \Theta')$, since:
\begin{align*}
(p\cdot \Delta+ (1-p)\cdot \Delta') (x)
&=p\cdot (\sum_{y\in X}\omega'(x,y)) + (1-p)\cdot (\sum_{y\in X}\omega''(x,y))\\
&=\sum_{y\in X}(p\cdot \omega'(x,y) + (1-p)\cdot \omega''(x,y))\\
&=\sum_{y\in X}\omega(x,y)
\end{align*}
and analogously we have
\[(p\cdot \Theta+ (1-p)\cdot \Theta' )(y)
%=p\cdot (\sum_{x\in X}\omega'(x,y)) + (1-p)\cdot (\sum_{x\in X}\omega''(x,y))
%=\sum_{x\in X}(p\cdot \omega'(x,y) + (1-p)\cdot \omega''(x,y))
=\sum_{x\in X}\omega(x,y).\]

Hence, for any pair of couplings $\omega'\in Coup(\Delta,\Theta), \omega''\in Coup(\Delta',\Theta')$,
\begin{align*}
&\kant(d) (p\cdot \Delta+ (1-p)\cdot \Delta', p\cdot \Theta+ (1-p)\cdot \Theta')\\
&= \inf_{\omega \in Coup(p\cdot \Delta+ (1-p)\cdot \Delta', p\cdot \Theta+ (1-p)\cdot \Theta')} 
\sum_{(x,y)}\omega(x,y) \cdot d(x,y)\\
&\leq  
\sum_{(x,y)}(p\cdot \omega' + (1-p)\cdot \omega'')(x,y) \cdot d(x,y)\\
&= 
p\cdot  
(\sum_{(x,y)}\omega'(x,y) \cdot d(x,y))
+ (1-p)\cdot 
(\sum_{(x,y)}\omega''(x,y) \cdot d(x,y))
\end{align*}
Then we derive the result:
\begin{align*}
&\kant(d) (p\cdot \Delta+ (1-p)\cdot \Delta', p\cdot \Theta+ (1-p)\cdot \Theta') \\
&\leq  
\inf_{{\omega'} \in {Coup(\Delta,\Theta)}} \inf_{{\omega''} \in {Coup(\Delta',\Theta')}} \Big(p\cdot  (\sum_{(x,y)}\omega'(x,y) \cdot d(x,y))
+ (1-p)\cdot 
(\sum_{(x,y)}\omega''(x,y) \cdot d(x,y))\Big)\\
&=
p\cdot (\inf_{{\omega'} \in {Coup(\Delta,\Theta)}} 
\sum_{(x,y)}\omega'(x,y) \cdot d(x,y))
+ (1-p)\cdot (\inf_{{\omega''} \in {Coup(\Delta',\Theta')}} 
\sum_{(x,y)}\omega''(x,y) \cdot d(x,y))\tag{by Proposition \ref{appendix:basic:infsup:properties}(INF)} \\
&= p\cdot \kant(d)(\Delta,\Theta)+ (1-p)\cdot \kant(d)(\Delta',\Theta')
\end{align*}
%\todo{Be sure the reference to INF at the end does not appear the main body of the paper because INF will be defined in appendix}
%(see property (INF) used in proof od Lemma \ref{lem:hausconv}).
\end{proof}

\begin{proof}[Proof of Lemma \ref{lem:hkconv}]
We need to show that for all $S,S',T,T' \in \cset (X)$ and for all $p \in (0,1)$ it holds
\[\hk (d) \big(\bindistrwms p S  T, \; \bindistrwms p {S'}  {T'}\big) \leq \bindistr p {\hk (d) (S,S')} {\hk (d) (T,T')}.\]

We first derive the inequality
\begin{equation}\label{hausconv1}
\begin{gathered}
p\cdot \Big(\sup_{\distr'\in S'} \big(\inf_{\distr \in S} \kant(d)(\distr,\distr') \big)\Big)+ (1-p)\cdot \Big(\sup_{\distrb'\in T'} \big(\inf_{\distrb \in T} \kant(d)(\distrb,\distrb') \big)\Big)\\
\geq \sup_{\Psi'\in \bindistrwms p {S'} {T'}} \big(\inf_{\Psi \in \bindistrwms p {S} {T}} \kant(d)(\Psi,\Psi')\big) 
  \end{gathered}
\end{equation}
as follows:
\begin{align*}
&p\cdot \Big(\sup_{\distr'\in S'} \big(\inf_{\distr \in S} \kant(d)(\distr,\distr') \big)\Big)+ (1-p)\cdot \Big(\sup_{\distrb'\in T'} \big(\inf_{\distrb \in T} \kant(d)(\distrb,\distrb') \big)\Big)\\
&= \sup_{\distr'\in S', \distrb'\in T'} \Big(p\cdot \big(\inf_{\distr \in S} \kant(d)(\distr,\distr')\big) + (1-p)\cdot \big(\inf_{\distrb \in T} \kant(d)(\distrb,\distrb')\big) \Big)
\tag{by Proposition \ref{appendix:basic:infsup:properties}(SUP)} \\
&= \sup_{\distr'\in S', \distrb'\in T'} \Big(\inf_{\distr \in S, \distrb\in T} \big(p\cdot \kant(d)(\distr,\distr') + (1-p)\cdot \kant(d)(\distrb,\distrb')\big) \Big)
\tag{by Proposition \ref{appendix:basic:infsup:properties}(INF)} \\
&\geq \sup_{\distr'\in S', \distrb'\in T'} \big(\inf_{\distr \in S, \distrb \in T} \kant(d)(\bindistr p \distr \distrb , \bindistr p {\distr'} {\distrb'})  \big)
 \tag{by Proposition \ref{prop:kantconv}
 %, generalized to arbitrary elements of $\dset(X)$, 
 and monotonicity of $\inf$ and $\sup$}\\
&= \sup_{\Psi'\in \bindistrwms p {S'} {T'}} \big(\inf_{\Psi \in \bindistrwms p {S} {T}} \kant(d)(\Psi,\Psi')\big)
 \tag{definition of $\wms$}\\
\end{align*}

We then use inequality (\ref{hausconv1}) to derive:
\begin{align*}
&\hk (d) \big(\bindistrwms p S  T,\; \bindistrwms p {S'}  {T'}\big)
\\
&=\max
 \Big\{
\sup_{\Psi'\in \bindistrwms p {S'} {T'}} \big(\inf_{\Psi \in \bindistrwms p {S} {T}} \kant(d)(\Psi,\Psi') \big) ,\\
&\qquad\qquad \sup_{\Psi'\in \bindistrwms p {S'} {T'}} \big(\inf_{\Psi \in \bindistrwms p {S} {T}} \kant(d)(\Psi,\Psi')\big)  
\Big\}
 \tag{definition of $\haus$}\\
 &\leq \max
 \Big\{
 p\cdot \Big(\sup_{\distr'\in S'} \big(\inf_{\distr \in S} \kant(d)(\distr,\distr') \big)\Big)
 + (1-p)\cdot \Big(\sup_{\distrb'\in T'} \big(\inf_{\distrb \in T} \kant(d)(\distrb,\distrb')\big) \Big),\\
&\qquad\qquad p\cdot \Big(\sup_{\distr\in S} \big(\inf_{\distr' \in S'} \kant(d)(\distr,\distr') \big)\Big)
 + (1-p)\cdot \Big(\sup_{\distrb\in T} \big(\inf_{\distrb' \in T'} \kant(d)(\distrb,\distrb')\big) \Big)
\Big\}
 \tag{by inequality (\ref{hausconv1})}
 \end{align*}
and since for any finite set $A,B$ it holds {$\max(A)\leq \max (A\cup B)$}, we have that this expression is less than or equal to
\begin{align*}
&\max 
 \Big\{
 p\cdot \Big(\sup_{\distr'\in S'} \big(\inf_{\distr \in S} \kant(d)(\distr,\distr') \big)\Big)
 + (1-p)\cdot \Big(\sup_{\distrb'\in T'} \big(\inf_{\distrb \in T} \kant(d)(\distrb,\distrb')\big) \Big),\\
&\qquad\qquad p\cdot \Big(\sup_{\distr'\in S'} \big(\inf_{\distr \in S} \kant(d)(\distr,\distr') \big)\Big)
 + (1-p)\cdot \Big(\sup_{\distrb\in T} \big(\inf_{\distrb' \in T'} d(\distrb,\distrb')\big) \Big),\\
&\qquad\qquad p\cdot \Big(\sup_{\distr\in S} \big(\inf_{\distr' \in S'} \kant(d)(\distr,\distr') \big)\Big)
 + (1-p)\cdot \Big(\sup_{\distrb'\in T'} \big(\inf_{\distrb \in T} \kant(d)(\distrb,\distrb')\big) \Big),\\
&\qquad\qquad p\cdot \Big(\sup_{\distr\in S} \big(\inf_{\distr' \in S'} \kant(d)(\distr,\distr') \big)\Big)
 + (1-p)\cdot \Big(\sup_{\distrb\in T} \big(\inf_{\distrb' \in T'} d(\distrb,\distrb')\big) \Big)
\Big\}
\end{align*}
which is in turn, by Proposition \ref{appendix:basic:infsup:properties}(SUP) applied to $\max$, equal to
\begin{align*}
&p\cdot  \max\Big\{\sup_{\distr'\in S'} \big(\inf_{\distr \in S} \kant(d)(\distr,\distr')\big), \sup_{\distr\in S} \big(\inf_{\distr' \in S'} \kant(d)(\distr,\distr')\big) \Big\}
+\\
&\quad (1-p) \cdot \max\Big\{\sup_{\distrb'\in T'} \big(\inf_{\distrb \in T} \kant(d)(\distrb,\distrb') \big), \sup_{\distrb\in T} \big(\inf_{\distrb' \in T'} \kant(d)(\distrb,\distrb')\big)\Big\} 
\\
&= \bindistr p {\hk (d) (S,S')} {\hk (d) (T,T')}  \tag{definition of $\haus$}
\end{align*}

\end{proof}

\begin{proof}[Proof of Lemma \ref{lem:newms}.]
%Let  $\minsum_{X}: \dset \pset X \to \pset \dset X$ be the weighted Minkowski sum, that is,
%$\minsum_{X} (\sum_{i} p_{i}  S_{i})= \{\sum_{i} p_{i}  x_{i} | \,x_{i} \in S_{i}\}$. 
%We prove, if $X$ is a convex space and $S_{i},T_{j}$ are convex for all $i,j$, then for any pair $\sum_{i} p_{i}  S_{i}, \sum_{j} q_{j} T_{j} \in \dset\pset X$ of distributions over convex sets it holds:
%We first prove the following inequation, for $\sum_{i} p_{i}  S_{i}, \sum_{j} q_{j} T_{j}$ probability distributions over sets $S_{i},T_{j}\in \cset (X)$.
%\begin{equation}\label{eq:hkh}
%\hk (d) (\sum_{i} p_{i} \cdot S_{i}, \sum_{j} q_{j} \cdot T_{j})
%\leq
%\kant\hk (d) (\sum_{i} p_{i}  S_{i}, \sum_{j} q_{j} T_{j})
%\end{equation}
%with $\sum_{i} p_{i} \pdot S_{i}$ denoting the weighted Minkowski sum $\{\sum_{i} p_{i} \Delta_{i} |\, \Delta_{i} \in S_{i}\} \in \cset (\dset\dset)$ and with $\mu^{\dset}$ the multiplication of the distribution monad.
Let $\sum_{i} p_{i}  S_{i}, \sum_{j} q_{j} T_{j}$ be probability distributions over sets $S_{i},T_{j}\in \cset (X)$.
Let $\omega \in Coup(\sum_{i} p_{i}  S_{i}, \sum_{j} q_{j}  T_{j})$ be an arbitrary coupling for such distributions (see Definition \ref{def:kantorovich:lifting}). Let $r_{i,j}\in[0,1]$ be defined as $r_{i,j}=\omega(S_{i},T_{j})$, for each $i\in I$ and $j\in J$.

Then:
\begin{align*}
&\hk (d) (
\wms(\sum_{i} p_{i} S_{i}), 
\wms(\sum_{j} q_{j}  T_{j}))\\
&=\hk (d) ( 
\wms(\sum_{i} (\sum_{j} r_{i,j} )S_{i}), 
\wms(\sum_{j} (\sum_{i} r_{i,j}) T_{j}))
& \tag{by $\omega$ being a coupling}\\
%&=\hk d ( 
%\pset \mu^{\dset}_{X} \circ\minsum_{\dset X} (\sum_{i} (\sum_{j} \omega(S_{i},T_{j}) S_{i})), 
%\pset \mu^{\dset}_{X} \circ\minsum_{\dset X} (\sum_{j} (\sum_{i} \omega(S_{i},T_{j}) T_{j})))
%& \tag{by convexity of $S_{i},T_{j}}\\
&=\hk (d) (
\wms(\sum_{(i,j)} r_{i,j}  S_{i}), 
\wms(\sum_{(i,j)} r_{i,j}T_{j}))
&\tag{by $S_{i},T_{j}$ convex for all $i,j$}\\
&\leq \sum_{(i,j)} r_{i,j} \cdot 
\hk (d) (S_{i}, T_{j})
& \tag{by Lemma \ref{lem:hkconv}, generalized to arbitrary finitely supported distributions}
%&= \sum_{<i,j>} \omega(S_{i},T_{j})
%\haus d (S_{i}, T_{j})
\end{align*}
Hence, we derive the result:
\begin{align*}
\hk (d) (
\wms(\sum_{i} p_{i}  S_{i}), 
\wms(\sum_{j} q_{j}  T_{j}))
&\leq \inf_{\omega \in Coup(\sum_{i} p_{i}  S_{i}, \sum_{j} q_{j}  T_{j})}
\sum_{(i,j)} r_{i,j}\cdot \hk (d) (S_{i}, T_{j}) & \\
&= \kant\hk (d) (\sum_{i} p_{i}  S_{i}, \sum_{j} q_{j} T_{j}) & \tag{by definition of $\kant$}\\
\end{align*}
\end{proof}

\begin{proof}[Proof of Proposition \ref{lem:hausmon}.]
\begin{align*}
\haus (d) (S,T)
&=
\max\{
\sup_{x\in S}\inf_{y\in T} d (x,y), 
\sup_{y\in T} \inf_{x\in S} d (x,y)\}\\
&\leq 
\max\{
\sup_{x\in S}\inf_{y\in T} d' (x,y), 
\sup_{y\in T} \inf_{x\in S} d' (x,y)\}
&\tag{by $d\funleq d'$ and by monotonicity of $\sup$ and $\inf$}\\
&= \haus (d') (S,T)
\end{align*}
\end{proof}

\begin{proof}[Proof of Proposition \ref{lem:hausfunpair}.]
Let $S,T$ be compact subsets of $(X,d)$. Then
\begin{align*}
\haus (d_{X})(S,T) 
&= 
\max\{
\sup_{x\in S}\inf_{x'\in T} d_{X} (x,x'), 
\sup_{x'\in T} \inf_{x\in S} d_{X} (x,x')\}\\
& =
\max\{
\sup_{x\in S}\inf_{x'\in T} d_{Y} (f(x),f(x')), 
\sup_{x'\in T} \inf_{x\in S} d_{Y} (f(x),f(x'))\}\\
&=\max\{
\sup_{y\in (\hs f)(S)}\inf_{y'\in (\hs f)(T)} d_{Y} (y,y'), 
\sup_{y'\in (\hs f)(T)} \inf_{y\in  (\hs f)(S)} d_{Y} (y,y')\}\\
&= \haus (d_{Y}) ({\hs f}(S), {\hs f}(T))
\end{align*}
\end{proof}

\section{Proofs for Section \ref{section:5:proof:presentation}}

We first recall some useful properties of the $\Sets$ monad $\cset$ and convex semilattices \cite{BSV19,BSV20ar}.

Propositions \ref{prop:ubf} and \ref{prop:ubmu} show how some operations in $\cset$ can be computed using unique bases.

\begin{proposition}[\cite{BSV20ar}, Lemma 5]\label{prop:ubf}
For $S\in \cset(X)$ and $f:X\to Y$, it holds
$\cset f(S) 
= \conv (\bigcup_{\distr \in \ub(S)} \{\dset f (\distr)\})$.
\end{proposition}

\begin{proposition}[\cite{BSV20ar}, Lemma 8]\label{prop:ubmu}
For $S\in \cset\cset(X)$, it holds
$$\mu(S)=\conv \big(\bigcup_{\ddistr \in \ub(S)}  \{ \sum_{T\in \support(\ddistr)} \ddistr(T) \cdot \distr_{T} \mid \textnormal{for each $T\in \support(\ddistr)$, $\distr_T\in \ub(T)$}\}\big).$$
%\wms(\sum_{T\in \support(\ddistr)} \ddistr(T) \ub(T))\]
%\ntodo{here I am using $\wms$ as an operation over arbitrary sets, not just convex!}
%\[
%\conv\bigcup_{\ddistr \in \ub(S)} \,\bigcup_{f\in \{f:\support(\ddistr)\to\dset(X) |\, f(T)\in \ub(T), \forall T\}} \,\{\sum_{T \in \support(\ddistr)}  \ddistr(T) \cdot f(T) \}\]
\end{proposition}
%\marginpar{porp. \ref{prop:ubmu} is used only in the proof for $\calg$ I think. The proof of the lemma is in the appendix (I put proposition cause maybe this be directly in the arxiv?)}

The following proposition generalizes the convexity equation $x \oplus y = x \oplus y \oplus (x \pplus p y)$ in the theory $\etcs$ of convex semilattices, by showing how we can derive in $\etcs$ that convex combinations of the base of a set can always be added to the set.

\begin{proposition}[\cite{BSV20ar}, Lemma 16]\label{prop:removingconv}
Let $\distr$ be a convex combination of $\bigcup_{i} \{\distr_i\}$ for $i$ ranging over a finite set. Then we can derive in the theory $\etcs$ of convex semilattices the equation:
$$\bigcplus_{i} (\bigpplus_{x \in \support(\distr_{i}) } \distr_{i}(x) \, x) = \big(\bigcplus_{i} (\bigpplus_{x \in \support(\distr_{i}) } \distr_{i}(x) \, x )\big) \cplus (\bigpplus_{x \in \support(\distr) } \distr(x) \, x)$$
\end{proposition}

\subsection{Proofs for section \ref{sec:F}}

\begin{lemma}\label{lem:axiomscs}
Let $((X,d),\alpha)$ be an object in $\EM(\lcset)$. The quantitative algebra $\calf((X,d),\alpha)= (X, \sigcs^{\alpha}, d)$ satisfies the quantitative inferences ($A$, $C$, $I$, $A_p$, $C_p$, $I_p$, $D$) of the theory $\qetcs$.
\end{lemma}
\begin{proof}[Proof of Lemma \ref{lem:axiomscs}.]
Since $d$ is a metric and thus assigns distance zero to the same elements of $X$, for all inferences of the form $\vdash t=_{0} s$ it is enough to prove that the terms $t$ and $s$ are interpreted in the algebra $\calf((X,d),\alpha)$ as the same elements of $X$.

Idempotency and commutativity of $\cplus$ are immediate as
\begin{align*}
&x \cplus^{\alpha} x = \alpha(\{\dirac x\}) = x  \tag{by $\alpha \circ\eta = id$}\\
&x_{1} \cplus^{\alpha} x_{2} = \alpha(\conv\{\dirac {x_{1}}, \dirac {x_{2}}\}) = x_{2} \cplus^{\alpha} x_{1}  
\end{align*}
For associativity of $\cplus$, we have
\begin{align*}
&(x_{1} \cplus^{\alpha} x_{2}) \cplus^{\alpha} x_{3}\\
&= \alpha ( \conv\{ \dirac{\alpha(\conv\{\dirac {x_{1}}, \dirac {x_{2}}\})},\dirac{x_{3}}\})  \\
&= \alpha ( \conv\{ \dirac{\alpha(\conv\{\dirac {x_{1}}, \dirac {x_{2}}\})},\dirac{\alpha(\{\dirac {x_{3}}\})}\}) &\tag{by $\alpha \circ\eta = id$}\\
&= \alpha \circ \cset\alpha ( \conv\{ \dirac{\conv\{\dirac {x_{1}}, \dirac {x_{2}}\}},\dirac{\{\dirac {x_{3}}\}}\}) 
\tag{by Proposition \ref{prop:ubf}}
\\
&= \alpha \circ \mu^{\cset} ( \conv\{ \dirac{\conv\{\dirac {x_{1}}, \dirac {x_{2}}\}},\dirac{\{\dirac {x_{3}}\}}\}) &\tag{by $\alpha \circ \cset \alpha = \alpha \circ \mu$}\\
&= \alpha ( \conv\{ \dirac {x_{1}}, \dirac {x_{2}}, \dirac {x_{3}}\}) &\tag{by definition of $\mu^{\cset}$ and Proposition  \ref{prop:ubmu}}
\end{align*}
and analogously we derive 
\begin{align*}
x_{1} \cplus^{\alpha} (x_{2} \cplus^{\alpha} x_{3})
&= \alpha \circ \mu^{\cset} ( \conv\{ \dirac{\{\dirac {x_{1}}\}},\dirac{\conv\{\dirac {x_{2}},\dirac {x_{3}}\}}\})\\
&= \alpha ( \conv\{ \dirac {x_{1}}, \dirac {x_{2}}, \dirac {x_{3}}\})
\end{align*}
which then concludes the proof.

For the axioms of convex algebras, we again have that idempotency and commutativity of $\pplus p$ are immediate as
\begin{align*}
&x \pplus p^{\alpha} x = \alpha(\{\dirac x\}) = x  &\tag{by $\alpha \circ\eta = id$}\\
&x_{1} \pplus p^{\alpha} x_{2} = \alpha(\{p {x_{1}} + (1-p)  {x_{2}}\}) = x_{2} \pplus {(1-p)}^{\alpha} x_{1}  
\end{align*}
Associativity of $\pplus p$ follows as that of $\cplus$. We have:
\begin{align*}
&(x_{1} \pplus q^{\alpha} x_{2}) \pplus p^{\alpha} x_{3}\\
&= \alpha \circ \cset\alpha ( \{p \{q x_{1} + (1-q){x_{2}}\}+(1-p) \{\dirac{x_{3}}\} \})\\
&= \alpha \circ \mu^{\cset} ( \{p \{q x_{1} + (1-q){x_{2}}\}+(1-p) \{\dirac{x_{3}}\}\} ) &\tag{by $\alpha \circ \cset \alpha = \alpha \circ \mu$}\\
&= \alpha ( \{ (pq) x_{1} + (p (1-q)){x_{2}} + (1-p) x_{3}\}) &\tag{by definition of $\mu^{\cset}$ and Proposition \ref{prop:ubmu}}
\end{align*}
and analogously we derive 
\begin{align*}
x_{1} \pplus {pq}^{\alpha} (x_{2} \pplus {\frac{p(1-q)}{1-pq}}^{\alpha} x_{3})
&= \alpha \circ \mu^{\cset} ( \{pq \{\dirac{x_{1}}\}+(1-pq) \{\frac{p(1-q)}{1-pq} x_{2} + (1-\frac{p(1-q)}{1-pq}){x_{3}}\} )\\
&= \alpha ( \{ (pq) x_{1} + (p (1-q)){x_{2}} + (1-p) x_{3}\})
\end{align*}

The distributivity axiom (D) follows as:
\begin{align*}
&x_{1} \pplus p^{\alpha} (x_{2} \cplus^{\alpha} x_{3}) \\
&= \alpha ( \{ p x_{1} + (1-p ) \alpha(\conv\{\dirac {x_{2}}, \dirac {x_{3}}\})\})  \\
&= \alpha ( \{ p \alpha(\{\dirac {x_{1}}\}) + (1-p ) \alpha(\conv\{\dirac {x_{2}}, \dirac {x_{3}}\})\})  &\tag{by $\alpha \circ\eta = id$}\\
&= \alpha \circ \cset\alpha ( \{ p \{\dirac {x_{1}}\} + (1-p ) \conv\{\dirac {x_{2}}, \dirac {x_{3}}\}\}) \\
&= \alpha \circ \mu^{\cset} ( \{ p \{\dirac {x_{1}}\} + (1-p ) \conv\{\dirac {x_{2}}, \dirac {x_{3}}\}\}) &\tag{by $\alpha \circ \cset \alpha = \alpha \circ \mu$}\\
&= \alpha ( \conv \{ p x_{1} + (1-p ) x_{2}, p x_{1} + (1-p ) x_{3}\}) &\tag{by definition of $\mu^{\cset}$ and Proposition \ref{prop:ubmu}}
\end{align*}
and analogously we derive 
\begin{align*}
(x_{1} \pplus p^{\alpha} x_{2}) \cplus (x_{1} \pplus p^{\alpha} x_{3})
&= \alpha \circ \mu^{\cset} ( \{ \dirac{\{p{x_{1}}+ (1-p ) x_{2}, p {x_{1}}+ (1-p ) {x_{3}}\}}\})\\
&= \alpha ( \conv \{ p x_{1} + (1-p ) x_{2}, p x_{1} + (1-p ) x_{3}\})
\end{align*}
\end{proof}

\newpage
The proof of Lemma \ref{non--expansive:lemma5.1.1} relies on the fact that 
the functions $\conv$, $\lambda x_{1},x_{2} . \{x_{1},x_{2}\}$, and $\lambda x_{1},x_{2}. (p x_{1}+(1-p) x_{2})$ are non-expansive. This is respectively proven in Lemmas \ref{lem:convne}, \ref{lem:pairsetne}, and \ref{lem:pairdistrne}
%\marginpar{changed main body to use these lemmas! Also, added proof of n.e. for $\lambda_{<d_1,d_2>}. \{p\cdot d_1 + (1-p)d_2 \}\big)$ used for (K)}

\begin{lemma}\label{lem:convne}
Let $(X,d)$ be a metric space. The function $\conv: \lfpset {\ldset(X,d)} \to \lcset(X,d)$ is non-expansive, i.e., 
for all $S,T\in \fpset {\dset(X)}$ it holds 
\[\hk (d) (\conv(S),\conv(T)) \leq \hk(d) (S,T) .\]
\end{lemma}
\begin{proof}
Let $S,T\in \fpset {\dset(X)}$. 
By the definition of Hausdorff metric, we want to prove that 
\begin{align*}
& \max\big\{  \sup_{\Delta \in \conv(S)}\inf_{\Theta \in \conv(T)}\kant(d)(\Delta, \Theta)    ,  \sup_{\Theta \in \conv(T)}\inf_{{\Delta \in \conv(S)}}\kant(d)(\Delta, \Theta)   \big\}\\
&\leq \max\big\{  \sup_{\Delta \in S}\inf_{\Theta \in T}\kant(d)(\Delta, \Theta)    ,  \sup_{\Theta \in T}\inf_{{\Delta \in S}}\kant(d)(\Delta, \Theta)   \big\}
\end{align*}
We show that
\[\sup_{\Delta \in \conv(S)}\inf_{\Theta \in \conv(T)}\kant(d)(\Delta, \Theta) \leq \sup_{\Delta \in S}\inf_{\Theta \in T}\kant(d)(\Delta, \Theta) \]
which then implies (by symmetry) the result.

Let $\Delta \in \conv(S)$. Then $\Delta$ is a convex combination of elements of $S$, that is, $\Delta = \sum_{i} p_{i}\cdot \Delta_{i}$ with $\Delta_{i}\in S$ and we have
\begin{align*}
\inf_{\Theta \in \conv(T)}\kant(d)(\Delta, \Theta) 
&\leq \inf_{\Theta \in T}\kant(d)(\Delta, \Theta) \tag{by $T\subseteq \conv(T)$}\\
&= \inf_{\Theta \in T}\kant(d)(\sum_{i} p_{i} \cdot\Delta_{i}, \Theta)\\
&\leq \inf_{\Theta \in T}\big(\sum_{i} p_{i} \cdot \kant(d)(\Delta_{i}, \Theta)\big)\tag{by Proposition \ref{prop:kantconv}}\\
&=\sum_{i} p_{i} \cdot \big ( \inf_{\Theta \in T} \kant(d)(\Delta_{i}, \Theta)\big)\tag{by Lemma \ref{appendix:basic:infsup:properties}(INF)}\\
&\leq \max_{i} \big\{ \inf_{\Theta \in T} \kant(d)(\Delta_{i}, \Theta)\big\}\\
&\leq \sup_{\Delta \in S} \inf_{\Theta \in T}\kant(d)(\Delta, \Theta)
\end{align*}
%\begin{equation}
%By $\ub(T)\subseteq T$ we have $$
%\end{equation}\label{eq:hkubleq}
%
%\begin{align*}
%\hk (d)() 
%&= \max\big\{  \sup_{i\in I}\inf_{i'\in I}d(x_i,y_{i'})    ,  \sup_{i'\in I}\inf_{i\in I}d(x_i,y_{i'})   \big\}\\
%\leq \sup_{i} d(x_{i},y_{i})
\end{proof}

\begin{lemma}\label{lem:pairsetne}
Let $(X,d)$ be a metric space and $d_{\sup}$ be the $\sup$--metric over $X\times X$. The function $\lambda  x_{1},x_{2}. \{x_{1},x_{2}\}: (X\times X, d_{sup}) \to \lfpset(X,d) $ is non-expansive, i.e.,
\[\haus (d)(\{x_{1},x_{2}\}, \{y_{1},y_{2}\})\leq \max\{ d(x_{1},y_{1}), d(x_{2},y_{2})\}\]
\end{lemma}
\begin{proof}
\vvcut
{Let $I=\{1,2\}$. Let $f: I \to I$ be a permutation of elements of the sets such that for each $i$, $d(x_{i},y_{f(i)})=\inf_{i'\in I}d(x_i,y_{i'})$.
Symmetrically, let $g: I\to I$ be a permutation such that $d(x_{g(i)},y_{i})=\inf_{i'\in I}d(x_{i'},y_{i})$.
Hence, for all $i,i',i'' \in I$ we have $d(x_{i},y_{f(i)}) \leq d(x_{i},y_{i'})$ and $d(x_{g(i)},y_{i}) \leq d(x_{i''},y_{i})$. This implies that for all $i\in I$, $d(x_{i},y_{f(i)}) \leq d(x_{i},y_{i})$ and $d(x_{g(i)},y_{i}) \leq d(x_{i},y_{i})$, which in turn implies that for all $i\in I$, $\inf_{i'} d(x_{i},y_{i'}) \leq d(x_{i},y_{i})$ and $\inf_{i'} d(x_{i'},y_{i}) \leq d(x_{i},y_{i})$. }
Let $I=\{1,2\}$.
For all $i\in I$, $\inf_{i'} d(x_{i},y_{i'}) \leq d(x_{i},y_{i})$ and $\inf_{i'} d(x_{i'},y_{i}) \leq d(x_{i},y_{i})$.
Hence, we derive:
\[\sup_{i} \inf_{i'} d(x_{i},y_{i'}) \leq \sup_{i} d(x_{i},y_{i})\quad \text{ and } \quad \sup_{i}\inf_{i'} d(x_{i'},y_{i}) \leq \sup_{i} d(x_{i},y_{i})\]
Then
\begin{align*}
\haus (d)(\{x_{1},x_{2}\}, \{y_{1},y_{2}\}) 
&= \max\big\{  \sup_{i}\inf_{i'}d(x_i,y_{i'})    ,  \sup_{i}\inf_{i'}d(x_{i'},y_{i})   \big\}\\
&\leq \sup_{i} d(x_{i},y_{i})\\
&= \max\{ d(x_{1},y_{1}), d(x_{2},y_{2})\}
\end{align*}
\end{proof}

\newpage
\begin{lemma}\label{lem:pairdistrne}
Let $(X,d)$ be a metric space and $d_{\sup}$ be the $\sup$--metric over $X\times X$. For every $p \in (0,1)$, the function $\lambda x_{1},x_{2}. (p x_{1}+(1-p) x_{2}): (X\times X, d_{sup}) \to \ldset(X,d) $ is non-expansive, i.e.,
\[\kant (d)(p x_{1}+(1-p) x_{2}, p y_{1}+(1-p) y_{2})\leq \max\{ d(x_{1},y_{1}), d(x_{2},y_{2})\}\]
\end{lemma}
\begin{proof}
By convexity of the Kantorovich metric, we have
\begin{align*}
&\kant (d)(p x_{1}+(1-p) x_{2}, p y_{1}+(1-p) y_{2}) \\
&=  \kant (d)(p \cdot \dirac{x_{1}}+(1-p) \cdot \dirac{x_{2}}, p \cdot\dirac{y_{1}}+(1-p) \cdot \dirac{y_{2}})\\
&\leq p \cdot  \kant (d)(\dirac{x_{1}}, \dirac{y_{1}}) +(1-p)\cdot \kant (d)(\dirac{x_{2}}, \dirac{y_{2}})\\
&= p \cdot  d(x_{1},y_{1}) +(1-p)\cdot d(x_{2},y_{2})
\end{align*}
Then we conclude by $p \cdot  d(x_{1},y_{1}) +(1-p)\cdot d(x_{2},y_{2}) \leq \max\{ d(x_{1},y_{1}), d(x_{2},y_{2})\}.$

%We first prove that for all $\distr,\distrb \in \dset(X)$ it holds
%\begin{equation}\label{eq:kanthaussupp}
%\kant(d) (\distr,\distrb) \leq \haus(d) (\support (\distr), \support(\distrb)).
%\end{equation}
%Then the result follows from equation (\ref{eq:kanthaussupp})  and Lemma \ref{lem:pairsetne} as
%\[\kant (d)(p x_{1}+(1-p) x_{2}, p y_{1}+(1-p) y_{2})\leq 
%\haus (d)(\{x_{1},x_{2}\}, \{y_{1},y_{2}\}) \leq
%\max\{ d(x_{1},y_{1}), d(x_{2},y_{2})\}\]
%To prove equation (\ref{eq:kanthaussupp}), it is enough
%to prove that for any coupling $\omega\in Coup(\distr,\distrb)$ it holds
%\[\sum_{{x\in \support (\distr), y \in \support (\distrb)}} \omega(x,y)\cdot d(x,y) \leq \haus(d) (\support (\distr), \support(\distrb)).\]
%We first observe that
%\begin{align*}
%&\inf_{x\in \support (\distr), y \in \support (\distrb)} d(x,y) \leq \sup_{x\in \support (\distr)} \inf_{y \in \support (\distrb)} d(x,y)\\
%&\inf_{x\in \support (\distr), y \in \support (\distrb)} d(x,y) \leq \sup_{y \in \support (\distrb)} \inf_{x\in \support (\distr)} d(x,y)
%\end{align*}
%and thus 
%\[\inf_{x\in \support (\distr), y \in \support (\distrb)} d(x,y) \leq 
%\max\{\sup_{x\in \support (\distr)} \inf_{y \in \support (\distrb)} d(x,y)\\
%&\inf_{x\in \support (\distr), y \in \support (\distrb)} d(x,y) \leq \sup_{y \in \support (\distrb)} \inf_{x\in \support (\distr)} d(x,y)\]
\end{proof}

\begin{lemma}
The functor $\calf$ is well-defined on morphisms.
\end{lemma}
\begin{proof}
We want to prove that whenever $f:((X,d), \alpha)\to ((Y,d'), \beta)$ is a non-expansive morphism of Eilenberg-Moore algebras then
$\calf(f)$ is a non-expansive homomorphism of convex semilattices. This means proving the following equations:
\begin{align*}
f(x_{1} \cplus^{\alpha} x_{2})
&= f({\alpha} (\conv\{\dirac {x_{1}}, \dirac {x_{2}}\}))\\
&= {\beta} \circ \cset f (\conv\{\dirac {x_{1}}, \dirac {x_{2}}\}) \tag{by $f$ a morphism of Eilenberg-Moore algebras}\\
&= {\beta}(\conv\{\dirac {f(x_{1})}, \dirac {f(x_{2})}\}) \tag{by Proposition \ref{prop:ubf}}\\
&= {f(x_{1})}\cplus^{\beta} {f(x_{2})}
\end{align*}
The equation $f(x_{1} \pplus p^{\alpha} x_{2})= {f(x_{1})}\pplus p^{\beta} {f(x_{2})}$ follows analogously.
\end{proof}

\subsection{Proofs for section \ref{sec:G}}

\begin{proof}[Proof of Lemma \ref{lem:Gwd}.1.]
The first equation immediately follows from the definition of $\alpha$:
\begin{align*}
&\alpha \circ \eta^{\cset} (x) =\alpha (\{\dirac x\})= (\nf (\{\dirac x\}))^{\alga} = x
\end{align*}
We are now left to prove 
\begin{align*}
&\alpha \circ \cset\alpha= \alpha \circ \mu^{\cset}
\end{align*}

In what follows, we often write $\nu(\Delta)$ to denote $\nu(\{\Delta\})$.

We first observe that in the theory of convex semilattices the following equation holds:
\begin{equation}\label{eq:nfhom}
\nf(\conv (S_{1} \cup S_{2}))= \nf(S_{1}) \oplus \nf(S_{2}) 
%\qquad \nf( \wms (p\, S_{1} + (1-p) \, S_{2}))= \nf(S_{1}) \pplus p \nf(S_{2})
\end{equation}
Indeed, we know from \cite{BSV20ar} 
\finaltodo{\marginpar{this is actually not explicit..!!!!}}
that the isomorphism $\ics$ satisfies:
$\ics(\conv (S_{1} \cup S_{2}))= [t_{1} \oplus t_{2}]_{/\etcs}$
for any $t_{1}\in \ics(S_{1})$ and $t_{2}\in \ics(S_{2})$. Hence, by $[\nf(\conv (S_{1} \cup S_{2}))]_{/\etcs}= \ics(\conv (S_{1} \cup S_{2}))$, and by $\nf(S_{1}) \in \ics(S_{1})$ and $\nf(S_{2}) \in \ics(S_{2})$,
we derive that $[\nf(\conv (S_{1} \cup S_{2}))]_{/\etcs}= [\nf(S_{1}) \oplus \nf(S_{2})]_{/\etcs}$.

For $S \in \cset\cset(X)$, we have
\begin{align*}
&\alpha \circ \cset\alpha (S) \\
&= \alpha(\conv(\bigcup_{\ddistr \in \ub(S)}\{\sum_{T \in \support(\ddistr)}  \ddistr(T) \alpha(T)\})) \tag{by Proposition \ref{prop:ubf}}\\
&= \big(\nf(\conv\bigcup_{\ddistr \in \ub(S)}\{\sum_{T \in \support(\ddistr)}  \ddistr(T) \alpha(T) \})\big)^{\alga}\tag{by definition of $\alpha$}\\
&= \big(\bigoplus_{\ddistr \in UB(S)} (\bigpplus_{T \in \support (\ddistr)} \ddistr(T) \alpha (T))\big)^{\alga} \tag{by definition of $\nf$}\\
&= \big(\bigoplus_{\ddistr \in UB(S)} (\bigpplus_{T \in \support (\ddistr)} \ddistr(T) (\nf (T))^{\alga})\big)^{\alga} \tag{by definition of $\alpha$}\\
&= \big(\bigoplus_{\ddistr \in UB(S)} (\bigpplus_{T \in \support (\ddistr)} \ddistr(T) \nf (T))\big)^{\alga} \tag{by definition of interpretation of a term in an algebra}\\
&= \Big(\bigoplus_{\ddistr \in UB(S)} \big(\bigpplus_{T \in \support (\ddistr)} \ddistr(T) (\bigoplus_{\distr\in \ub(T)} \nf(\distr))\big)\Big)^{\alga} \tag{by definition of $\nf$}
\end{align*}
On the other side, we have
\begin{align*}
&\alpha \circ \mu^{\cset} (S) \\
&= \alpha\big(\conv(\bigcup_{\ddistr \in \ub(S)} \,(\bigcup_{f\in \{f:\support(\ddistr)\to\dset(X) |\, f(T)\in \ub(T), \forall T\}} \,\{\sum_{T \in \support(\ddistr)}  \ddistr(T) f(T) \}))\big) \tag{by Proposition \ref{prop:ubmu}} \\
&= \Big(\nf\big(\conv(\bigcup_{\ddistr \in \ub(S)} \,(\bigcup_{f\in \{f:\support(\ddistr)\to\dset(X) |\, f(T)\in \ub(T), \forall T\}} \,\{\sum_{T \in \support(\ddistr)}  \ddistr(T) f(T) \}))\big)\Big)^{\alga} \tag{by definition of $\alpha$}\\
%&= \nf(\conv\bigcup_{\ddistr \in \ub(S)} \,\bigcup_{f\in \{f:\support(\ddistr)\to\dset(X) |\, f(T)\in \ub(T), \forall T\}} \,\{\sum_{T \in \support(\ddistr)}  \ddistr(T) f(T) \})^{\alga} \tag{by definition of $\nf$}\\
&= \Big(\nf\big(\conv(\bigcup_{(\ddistr \in \ub(S), f\in \{f:\support(\ddistr)\to\dset(X) |\, f(T)\in \ub(T), \forall T\})} \,\{\sum_{T \in \support(\ddistr)}  \ddistr(T) f(T) \})\big)\Big)^{\alga}\\
&= \Big(\bigoplus_{(\ddistr \in \ub(S), f\in \{f:\support(\ddistr)\to\dset(X) |\, f(T)\in \ub(T), \forall T\})} \,\nf(\sum_{T \in \support(\ddistr)}  \ddistr(T) f(T)) \Big)^{\alga} \tag{by (\ref{eq:nfhom})}\\
&= \Big(\bigoplus_{(\ddistr \in \ub(S), f\in \{f:\support(\ddistr)\to\dset(X) |\, f(T)\in \ub(T), \forall T\})} \,\big(\bigpplus_{T \in \support(\ddistr)} \ddistr(T) \, \nf(f(T))\big) \Big)^{\alga} \tag{by definition of $\nf$}
\end{align*}
Hence, we can conclude if we derive in the theory of convex semilattices that 
\begin{align*}
&\bigoplus_{\ddistr \in UB(S)} \bigpplus_{T \in \support (\ddistr)} \ddistr(T) (\bigoplus_{\distr\in \ub(T)} \nf(\distr)) \\
&= \bigoplus_{(\ddistr \in \ub(S), f\in \{f:\support(\ddistr)\to\dset(X) |\, f(T)\in \ub(T), \forall T\})} \,\big(\bigpplus_{T \in \support(\ddistr)} \ddistr(T) \, \nf(f(T))\big)
\end{align*}
as this guarantees that the terms will be interpreted as the same element of $X$ in the algebra $\alga$.
This is derived from the fact that for every $\ddistr \in UB(S)$ it holds
\begin{align*}
&\bigpplus_{T \in \support (\ddistr)} \ddistr(T) (\bigoplus_{\distr\in \ub(T)} \nf(\distr)) 
= \bigoplus_{f\in \{f:\support(\ddistr)\to\dset(X) |\, f(T)\in \ub(T), \forall T\}} \,\big(\bigpplus_{T \in \support(\ddistr)} \ddistr(T) \, \nf(f(T))\big)
\end{align*}
which is an instance of the generalized version of axiom (D).
Indeed, by iterating the distributivity axiom (D), we derive in the theory of convex semilattices that:
\[\bigpplus_{1\leq i\leq k} p_{i}\, (t^{i}_{1}\oplus ...\oplus t^{i}_{n_{i}}) =
\bigoplus_{(t_{1},...,t_{k})\in \{(t_{1},...,t_{k})|\, t_{i} \in \{t^{i}_{1},...,t^{i}_{n_{i}}\}\}} (\bigpplus_{1\leq i\leq k} p_{i} \,t_{i})\]
and this law can be alternatively written as follows, whenever for each $i$ we have a set of terms $S_{i}$:
\[\bigpplus_{1\leq i\leq k} p_{i}\, (\bigoplus _{t\in S_{i}} t)= 
\bigoplus_{f\in \{f: \{1,...,k\} \to \terms X \sigcs |\, f({i}) \in S_{i}\}} (\bigpplus_{1\leq i\leq k} p_{i}\, f(i))\]
where $\{f: \{1,...,k\} \to \terms X \sigcs |\, f({i}) \in S_{i}\}$ is the set of functions choosing one term in each $S_{i}$.
\end{proof}

We show that in the theory $\qetcs$ the Kantorovich distance of two distributions is an upper bound to the distance of their corresponding terms given by $\nf$.
\begin{lemma}\label{lem:Gkant}
Let $(X,d)$ be a metric space and let $\Delta, \Theta \in \dset(X)$. Then
\[\bigcup_{(x,y)\in \support (\Delta) \times \support (\Theta)}\{x=_{d(x,y)} y\}\vdash  \nf(\{\Delta\}) =_{\kant(d)(\Delta,\Theta)} \nf(\{\Theta\}).\]
\end{lemma}
\begin{proof}
Let $\Delta=\sum_{i} p_{i} x_{i} $ and $\Theta= \sum_{j} q_{j} y_{j}$.
We have \[\kant(d)(\Delta,\Theta)=\inf_{\omega\in Coup(\Delta,\Theta)} (\sum_{(i,j)} {\omega(x_{i},y_{j})} \cdot {d(x_{i},y_{j})}).\] 
As the supports of the distributions are finite, there is some minimal coupling $\omega$ such that 
\[\sum_{(i,j)} \omega(x_{i},y_{j}) \cdot d{(x_{i},y_{j})}= \inf_{\omega\in Coup(\Delta,\Theta)} \big(\sum_{(i,j)} {\omega(x_{i},y_{j})} \cdot {d(x_{i},y_{j})}\big)\]
Hence, it is enough to prove that for any coupling $\omega\in Coup(\Delta,\Theta)$ we have
\[\{x_{i}=_{d(x_{i},y_{i})} y_{j}\}_{(i,j)}\vdash 
 \bigpplus_{i} p_{i} \,x_{i}  =_{\sum_{(i,j)} \omega(x_{i},y_{j}) \cdot d(x_{i},y_{j})} \bigpplus_{j} q_{j} \,y_{j}.\]
 Let $\Gamma$ be the set of hypothesis $\{x_{i}=_{d(x_{i},y_{i})} y_{j}\}_{(i,j)}$.
%By the Kantorovich duality, 
% the Kantorovich distance between the distributions is given by the infimum of the probabilistic couplings, i.e., 
%% there is a set of probability values $r_{k}$, for $k \in <i,j>$, such that 
%\[\kant(d)(\Delta,\Theta)= \inf_{\omega\in Coup(\Delta,\Theta)} \sum_{(x,y)\in \support (\omega)} {\omega(x,y)} {d(x,y)} \]
Let $\omega$ be a coupling for $\Delta,\Theta$.
%%not sure whether here I should specify that all x_{i} are different
By rule (K) we derive:
%%%here: solve problem below that $\multpplus j {r_{<i,j>}} {x_{i}}$ is possibly not a full distirbution!!
\begin{equation}\label{eq:Gkantomega}
\Gamma \vdash
\bigpplus_{(i,j)} {\omega(x_{i},y_{j})} \,{x_{i}} 
=_{\sum_{(i,j)} \omega(x_{i},y_{j}) \cdot \epsilon_{(i,j)}} 
\bigpplus_{(i,j)} {\omega(x_{i},y_{j})} \,{y_{j}}
\end{equation}
As $\omega$ is a coupling for the distributions $(\Delta,\Theta)$, $\sum_{j} \omega(x_{i},y_{j})= p_{i}$ for every $i$.
Using the convex algebra axioms, it is easy to see that  
\[\vdash \bigpplus_{i} p_{i} \,x_{i} =_{0} \bigpplus_{(i,j)} {\omega(x_{i},y_{j})} \,{x_{i}}\]
Analogously, by $\sum_{i} \omega(x_{i},y_{j})= \sum_{i} p_{i}$ for every $j$ we derive
\[\vdash \bigpplus_{j} q_{j} \,y_{j} =_{0} \bigpplus_{(i,j)} {\omega(x_{i},y_{j})} \,{y_{j}}\]
%TO DO?
Hence, we derive from (\ref{eq:Gkantomega}), using the inference rules (Triang) and (Cut), that:
\[\Gamma\vdash
\bigpplus_{i} p_{i} \,x_{i} 
=_{\sum_{(i,j)} \omega(x_{i},y_{j}) \cdot d(x_{i},y_{j})} 
\bigpplus_{j} q_{j} \,y_{j}
\]
\end{proof}

\begin{proof}[Proof of Lemma \ref{lem:Ghk}.]
Let $(X,d)$ be a metric space and let $S,T \in \cset(X)$. We want to prove that in $\qetcs$ it holds:
\[
\bigcup_{(\Delta,\Theta)\in\ub(S)\times \ub(T)} (\bigcup_{(x,y) \in \support(\Delta)\times \support(\Theta)}\{x=_{d(x,y)} y\})
\vdash 
\nf(S) 
%\bigcplus_{(\sum_{i} p_{i} x_{i})\in \ub(S) }^{\alga} \big(\bigpplus_{i}^{\alga} p_{i} \alpha(x_{i})\big)
=_{\hk(d)(S,T)}  
\nf(T)
\]
In what follows, we often write $\nu(\Delta)$ to denote $\nu(\{\Delta\})$.

By definition of $\nf$, we have 
\[
\nf(S) =
\bigcplus_{\Delta\in \ub(S) } \nf(\Delta)
\quad\quad\quad
\nf(T) =
\bigcplus_{\Theta\in \ub(T) } \nf(\Theta)
%\nf(S) =
%\bigcplus_{(\sum_{i} p_{i} x_{i})\in \ub(S) } (\bigpplus_{i} p_{i} x_{i})
%\quad\quad\quad
%\nf(T) =
%\bigcplus_{(\sum_{j} q_{j} x_{j})\in \ub(T) } (\bigpplus_{j} q_{j} y_{j})
\]
By the definition of the Hausdorff lifting, 
for each $\Delta \in S$ there is a $\Theta_{\Delta} \in T$ such that
\begin{equation}\label{eq:Ghk01}
\kant(d)(\Delta,\Theta_{\Delta})\leq \haus\kant(d) (S,T)
\end{equation}
and analogously for each $\Theta \in T$ there is a $\Delta_{\Theta} \in S$ such that
\begin{equation}\label{eq:Ghk02}
\kant(d)(\Delta_{\Theta},\Theta)\leq \haus\kant(d) (S,T)
\end{equation}

For each $\Theta \in \ub(T)$,  since $\Delta_{\Theta} \in S$ then $\Delta_{\Theta}$ is a convex combination of elements of $\ub(S)$. Then by Proposition \ref{prop:removingconv} we derive that for every $\Theta \in \ub(T)$,
\[\vdash \bigcplus_{\Delta\in \ub(S)} \nf(\Delta) =_{0} \big(\bigcplus_{\Delta\in \ub(S)} \nf(\Delta)\big) \cplus \nf(\Delta_{\Theta})\]
and thus by multiple applications of (Triang) and (Cut) we have
\begin{equation}\label{eq:Ghk11}
\vdash \bigcplus_{\Delta\in \ub(S)} \nf(\Delta) =_{0} \big(\bigcplus_{\Delta\in \ub(S)} \nf(\Delta)\big) \cplus \big(\bigcplus_{\Theta \in \ub(T)} \nf(\Delta_{\Theta})\big)
\end{equation}
Symmetrically, we derive
\begin{equation}\label{eq:Ghk12}
\vdash \bigcplus_{\Theta\in \ub(T)} \nf(\Theta) =_{0} \big(\bigcplus_{\Theta\in \ub(T)} \nf(\Theta)\big) \cplus \big(\bigcplus_{\Delta \in \ub(S)} \nf(\Theta_{\Delta})\big)
\end{equation}

By Lemma \ref{lem:Gkant} for every $\Delta\in \ub(S)$ we have
\[\bigcup_{(x,y)\in \support (\Delta) \times \support (\Theta_{\Delta})}\{x=_{d(x,y)} y\}\vdash  \nf(\Delta) =_{\kant(d)(\Delta,\Theta_{\Delta})} \nf(\Theta_{\Delta}).\]
Hence, by (\ref{eq:Ghk01}) and rule (Max), for every $\Delta \in \ub(S)$ we derive:
\begin{equation}\label{eq:Ghk000}
\bigcup_{(x,y)\in \support (\Delta) \times \support (\Theta_{\Delta})}\{x=_{d(x,y)} y\}\vdash  \nf(\Delta) =_{\hk(d)(S,T)} \nf(\Theta_{\Delta}).
\end{equation}
Now, define the set of hypothesis
\[\Gamma=\bigcup_{\Delta\in \ub(S),\Theta\in\ub(T)} (\bigcup_{(x,y) \in \support(\Delta)\times \support(\Theta)}\{x=_{d(x,y)} y\})\]
As $\Theta_{\Delta}$ is a convex combination of elements of $\ub(T)$, the elements in its support are included in $\bigcup_{\Theta \in \ub (T)} \support(\Theta)$. Hence, 
by (\ref{eq:Ghk000}) and rules (Cut) and (Assum) we derive that for every $\Delta\in \ub(S)$,
\begin{equation}\label{eq:Ghk21}
\Gamma \vdash  \nf(\Delta) =_{\hk(d) (S,T)} \nf(\Theta_{\Delta}).
\end{equation}

Symmetrically from Lemma \ref{lem:Gkant}, by (\ref{eq:Ghk02}) and by rules (Max),(Cut), and (Assum), we derive that for every $\Theta \in \ub(T)$:
\begin{equation}\label{eq:Ghk22}
\Gamma\vdash  \nf(\Delta_{\Theta}) =_{\hk(d) (S,T)} \nf(\Theta)
\end{equation}

From (\ref{eq:Ghk21}) and (\ref{eq:Ghk22}),
by multiple applications of (H), together with rules (Cut) and (Assum) to reach the set of hypothesis $\Gamma$, we derive:
\begin{equation}\label{eq:Ghk3}
\Gamma\vdash \big(\bigcplus_{\Delta\in \ub(S)} \nf(\Delta)\big) \cplus \big(\bigcplus_{\Theta \in \ub(T)} \nf(\Delta_{\Theta}) \big)
=_{\hk(d) (S,T)}  
\big(\bigcplus_{\Delta \in \ub(S)} \nf(\Theta_{\Delta})\big) \cplus \big(\bigcplus_{\Theta\in \ub(T)} \nf(\Theta)\big)
\end{equation}
Thus, by (\ref{eq:Ghk3}), (\ref{eq:Ghk11}), and (\ref{eq:Ghk12}), using rule (Triang) and commutativity of $\cplus$, together with rules (Cut) and (Assum) to reach the set of hypothesis $\Gamma$, we conclude:
\begin{equation}\label{eq:Ghk2020}
\Gamma\vdash \bigcplus_{\Delta\in \ub(S)} \nf(\Delta)
=_{\hk(d) (S,T)}  
\bigcplus_{\Theta\in \ub(T)} \nf(\Theta)
\end{equation}
\end{proof}

\begin{lemma}
The functor $\calg$ is well-defined on morphisms.
\end{lemma}
\begin{proof}
Let $f:X\to Y$ be a non-expansive homomorphism between the quantitative algebras $\alga =(X, \sigcs^{\alga},d)$
and $\algb =(Y, \sigcs^{\algb},d')$ in $\qacat(\qetcs)$.
Then $f$ is an arrow in $\Met$, being non-expansive. We now show that $f$ is also a morphism of Eilenberg-Moore algebras. We first observe that
\begin{align*}
f \circ \alpha (S) 
&= f ((\nf(S))^{\alga})
\tag{definition of $\alpha$}\\
&= f (\bigoplus_{\distr \in \ub(S)}^{\alga} (\bigpplus_{x \in \support (\distr)}^{\alga} \distr(x) x))
\tag{definition of $\nf$ and of interpretation in $\alga$}\\
& = \bigoplus_{\distr \in \ub(S)}^{\algb} (\bigpplus_{x \in \support (\distr)}^{\algb} \distr(x) f(x))
\tag{by $f$ an homomorphism}\\
& = \big(\bigoplus_{\distr \in \ub(S)} (\bigpplus_{x \in \support (\distr)} \distr(x) f(x))\big)^{\algb}
\end{align*}
By definition we have
\[\cset f(S)= \conv (\bigcup_{\distrb \in \ub(\cset f(S))} \{\distrb\})\]
By Proposition \ref{prop:ubf}, 
$\bigcup_{\distr \in \ub(S)} \{\dset f (\distr)\}$ is also a base for $\cset f (S)$, although possibly it is not the unique, minimal base. This means, that $\bigcup_{\distr \in \ub(S)} \{\dset f (\distr)\}$ contains $\ub(\cset f(S))$, and the remaining elements of $\bigcup_{\distr \in \ub(S)} \{\dset f (\distr)\}$ are convex combinations of $\ub(\cset f(S))$.
By Proposition \ref{prop:removingconv}, we can then derive in the theory of convex semilattices that
\[\bigoplus_{\distr \in \ub(S)} (\bigpplus_{x \in \support (\distr)} \distr(x) f(x)) = \bigoplus_{\distrb \in \ub(\cset f(S))} (\bigpplus_{y \in \support (\distrb)} \distrb(y) y)\]
As $\algb$ is a quantitative algebra for the theory $\qetcs$, the interpretation of such terms in  $\algb$ will be the same, and so we conclude:
\begin{align*}
\big(\bigoplus_{\distr \in \ub(S)} (\bigpplus_{x \in \support (\distr)} \distr(x) f(x)\big)^{\algb}
&= \big(\bigoplus_{\distrb \in \ub(\cset f(S))} (\bigpplus_{y \in \support (\distrb)} \distrb(y) y)\big)^{\algb}\\
&= (\nf(\cset f(S)))^{\algb}\\
&= \beta \circ \cset f (S)
\end{align*}
\end{proof}

\subsection{Proofs for Section \ref{sec:iso}}

We prove that the equations 
\[\calg \circ \calf = id_{\EM(\lcset)}\qquad \calf \circ \calg = id_{\qacat(\qetcs)}\]
hold for objects of the categories.

For $\calg \circ \calf = id_{\EM(\lcset)}$, let $((X,d),\alpha)$ be an object in $\EM(\lcset)$. Then
\begin{align*}
\calg \circ \calf ((X,d),\alpha)
= \calg(\alga)
= ((X,d),\alpha')
\end{align*}
with $\alga= (X, \sigcs^{\alpha}, d)$ defined accordingly to the definition of $\calf$,
and with $\alpha'(S)= (\nf(S))^{\alga}$ for any $S\in \cset (X)$.
We prove that $\alpha=\alpha'$.
For $S\in \cset (X)$, we have
\begin{align*}
\alpha'(S)
&= (\nf(S))^{\alga}
\tag{definition of $\alpha'$}\\
&= \big(\bigoplus_{\distr \in \ub(S)}(\bigpplus_{x\in \support (\distr)} \distr(x) \,x )\big)^{\alga}
\tag{definition of $\nf$}\\
&= \bigoplus_{\distr \in \ub(S)}^{\alga} (\bigpplus_{x\in \support (\distr)}^{\alga} \distr(x) \,x )
\tag{definition of interpretation in an algebra}\\
&= \alpha\Big(\conv \big(\bigcup_{\distr \in \ub(S)} \{\dirac {\alpha(\{\distr\})}\}\big)\Big)
\tag{definition of $\alga$}\\
&= \alpha \circ \cset \alpha \Big(\conv \big(\bigcup_{\distr \in \ub(S)} \{\dirac {\{\distr\}} \}\big)\Big)
\tag{Proposition \ref{prop:ubf}}\\
&= \alpha \circ \mu^{\cset} \Big(\conv \big(\bigcup_{\distr \in \ub(S)} \{\dirac {\{\distr\}} \}\big)\Big)
\tag{by $\alpha \circ \cset \alpha = \alpha \circ \mu^{\cset}$}\\
&= \alpha \Big(\conv \big(\bigcup_{\distr \in \ub(S)} \{\distr\}\big)\Big)
\tag{by definition of $\mu^{\cset}$}\\
&= \alpha (S)
\end{align*}

It remains to prove the second equation on objects, that is, $\calf \circ \calg = id_{\qacat(\qetcs)}$. Let $\alga= (X, \sigcs^{\alga}, d)$ be a quantitative algebra in $\qacat(\qetcs)$. We have
\begin{align*}
\calf \circ \calg (\alga)
= \calf((X,d), \alpha)
= \alga'
\end{align*}
with $\alpha$ defined accordingly to the definition of $\calg$ and 
$\alga'=(X, \sigcs^{\alpha}, d)$.
We want to prove that the interpretation of the convex semilattice operations in the algebras $\alga$ and $\alga'$ coincide.
We have 
\begin{align*}
x \cplus^{\alpha} y 
&= \alpha(\conv \{\dirac {x} , \dirac{y}\}) \\
&= (\nf(\conv \{\dirac {x} , \dirac{y}\}))^{\alga}\\
&= (x \oplus y)^{\alga}\\
&= x \oplus^{\alga} y
\end{align*}
and analogously we derive
\begin{align*}
x \pplus p^{\alpha} y 
&= \alpha(\{p \, x + (1-p) \, y\}) \\
&= (\nf(\{p \, x + (1-p) \, y\}))^{\alga}\\
&= (x \pplus p y)^{\alga}\\
&= x \pplus p^{\alga} y
\end{align*}
This last two proof are modulo the ordering of the elements in the set and in the support of the distributions as given by $\nf$, which is however irrelevant as, being $\alga$ a quantitative algebra for $\qetcs$, terms with a different ordering will have the same interpretation.

\end{document}